\documentclass[journal]{IEEEtran}
\usepackage[numbers,sort&compress,square]{natbib}
\usepackage[dvips]{graphicx}      
\usepackage{epsfig}               
\usepackage{color}
\usepackage{url}
\usepackage[T1]{fontenc}
\usepackage{amssymb,amsmath,amsfonts}     
\usepackage{times}
\usepackage{bm} 
\usepackage{amsbsy}
\usepackage{xy}
\usepackage{graphics}
\usepackage{mathrsfs}
\usepackage{nicefrac}
\usepackage[mathscr]{eucal}

\usepackage{algorithmicx}
\usepackage[ruled]{algorithm}
\usepackage{algpseudocode}
\graphicspath{ {./figures/} }
\usepackage{pgfplots}
\usepgfplotslibrary{colormaps}
\pgfplotsset{
    compat=newest,
	colormap={justblackandwhite}{color=(white) color=(white) color=(white)}
}

\usepackage{amsthm}
\newtheorem{lemma}{Lemma}

\title{
Covariance Fitting Interferometric Phase Linking: Modular Framework and Optimization Algorithms
}
\date{January 2024}
\author{Phan Viet Hoa Vu, Arnaud Breloy, Frédéric Brigui, Yajing Yan {\it Member IEEE}, and Guillaume Ginolhac {\it Senior Member, IEEE}
\thanks{
P. V. H. Vu is with ONERA-DEMR, University Paris Saclay, and with LISTIC (EA3703), University Savoie Mont-Blanc.
A. Breloy is with CEDRIC, CNAM, Paris.
F. Brigui is with ONERA-DEMR, University Paris Saclay.
Y. Yan and G. Ginolhac are with LISTIC (EA3703), University Savoie Mont-Blanc.
}
}

\begin{document}

\maketitle

\begin{abstract}
Interferometric phase linking (IPL) has become a prominent technique for processing images of areas containing distributed scaterrers in SAR interferometry.
Traditionally, IPL consists in estimating consistent phase differences between all pairs of SAR images in a time series from the sample covariance matrix of pixel patches on a sliding window.
This paper reformulates this task as a covariance fitting problem: in this setup, IPL appears as a form of projection of an input covariance matrix so that it satisfies the phase closure property.
Given this modular formulation, we propose an overview of covariance matrix estimates, regularization options, and matrix distances, that can be of interest when processing multi-temporal SAR data.
In particular, we will observe that most of the existing IPL algorithms appear as special instances of this framework.
We then present tools to efficiently solve related optimization problems on the torus of phase-only complex vectors: majorization-minimization and Riemannian optimization.
We conclude by illustrating the merits of different options on a real-world case study.
\end{abstract}

\section{Introduction}

Recent remote sensing missions (Sentinel-1, UAVSAR, TerraSAR-X, etc.) have brought an unprecedented amount of available synthetic aperture radar (SAR) images time series.
For interferometric SAR (InSAR), these systematic and regular acquisitions enabled the utilization of multi-temporal techniques (MT-InSAR), which significantly enhanced the accuracy of Earth displacement estimation.
Under the assumption of distributed targets, interferometric phase-linking (IPL) has emerged as a fundamental methodology to estimate phase differences from all available SAR acquisitions within the dataset \cite{Guarnieri2008, Ferretti2011_squee, Fornaro2015, Cao2015, jiang2020distributed, ansari2018efficient, ho2022compressed, zwieback2022cheap, vu2023covariance, bai2023lamie}.
The driving idea behind this technique is to leverage the redundancy of the time series
in order to compensate for the coherence loss between images over time, as illustrated in Fig.~\ref{fig:IPL_principle}.
The initial formulation of IPL was obtained from the perspective of approximate maximum likelihood estimation assuming a complex circular Gaussian model.
This estimation procedure then appears as an optimization problem aiming to retrieve phase estimates from the sample covariance matrix of each pixel patch.
Subsequent works motivated many variants, e.g., by improving coherence pre-estimation step \cite{jiang2020distributed}, using compression schemes \cite{ho2022compressed}, relaxing the optimization problem \cite{ansari2018efficient}, or assuming non-Gaussian models \cite{schmitt2014adaptive, wang2015robust}.
An overview of advances in this scope is presented in \cite{minh2023interferometric}.

In essence, all IPL algorithms aim to recover a property called phase closure (referring to the continuity of phase differences over the time series) within the sample covariance matrix.
In this paper, we leverage this reinterpretation to reformulate IPL as a covariance fitting problem \cite{Ottersten1998, Shahbazpanahi2004, Yardibi2010, Stoica2011, Meriaux2017, Mériaux2019, Mériaux2019b,Mériaux2021}: the task is thus expressed as fitting constrained phases to the modulus of any plug-in estimate of the covariance matrix.
This approach, referred to as COFI-PL, offers a concise and modular formulation.
It encompasses the maximum likelihood approaches as a special case, as well as many other generalizations concerning the construction of the covariance matrix plug-in estimate, and the fitting objective function.
In this scope, we overview staple building blocks from the state-of-the-art and relate them to existing IPL algorithms.
Furthermore, we consider new options bringing promising results on a real-world case study, notably the use of least-squares fitting objective, and the phase-only sample correlation matrix.

As all of the considered methods lead to the construction of an optimization problem over the torus of phase-only complex vectors $\mathbb{T}_p$, we further investigate two optimization frameworks suited to this space.
Contrarily to the majority of existing works that focus on optimizing the phases directly, this offers several advantages: it simplifies tedious computations of trigonometric functions to simpler matrix operations, and enables us to fully harness the geometric structure of the constrained space (for example, invariance of the phases modulo $2\pi$ is not an appearing issue when considering $\mathbb{T}_p$ directly).
The first framework is majorization-minimization \cite{hunter2004tutorial, sun2016majorization} on $\mathbb{T}_p$.
Its use is limited to quadratic problems on $\mathbb{T}_p$, which interestingly includes classical maximum likelihood \cite{Vu2022igarss, Vu2023robust} and least-squares based \cite{vu2023covariance} IPL formulations.
The main appeal of this framework lies in its ability to yield cost-efficient algorithms (here similar to modified power-methods), with guaranteed monotonicity, and with no required hyperparameter such as a step-size selection.
Hence, they appear suited to large-scale implementation.
The second framework is Riemannian optimization on $\mathbb{T}_p$ \cite{smith1999optimum, elmossallamy2021ris, xiong2023mimo}.
It is more flexible as it allows to work with any objective function.
It also opens the path to acceleration schemes and generalizations from many algorithms of the Riemannian optimization framework \cite{absil2008optimization, boumal2023introduction}.

The rest of the paper is organized as follows:
Section \ref{sec:background} presents IPL and the corresponding covariance matrix structure.
Section \ref{sec:COFIPL_framework} presents the generic construction of a covariance fitting IPL (COFI-PL) optimization problem.
This framework involves three construction modules: some design options suited to SAR data are respectively presented in Sections \ref{sec:cm_est} (covariance matrix estimation),
Section \ref{sec:cm_regul} (covariance matrix regularization), and Section \ref{sec:matrix_dist} (matrix distance objective function).
Section \ref{sec:mm} and Section \ref{sec:riem_opt} respectively introduce the majorization-minimization and Riemannian optimization on $ \mathbb{T}_p$ to solve
the resulting COFI-PL problems.
Section \ref{sec:sota} overviews existing IPL algorithms and links them to COFI-PL formulations.
Section \ref{sec:simu} presents an application of the proposed methods to a real-world case study.

\begin{figure}[!t]
    \centering
    \includegraphics[width=0.8\linewidth]{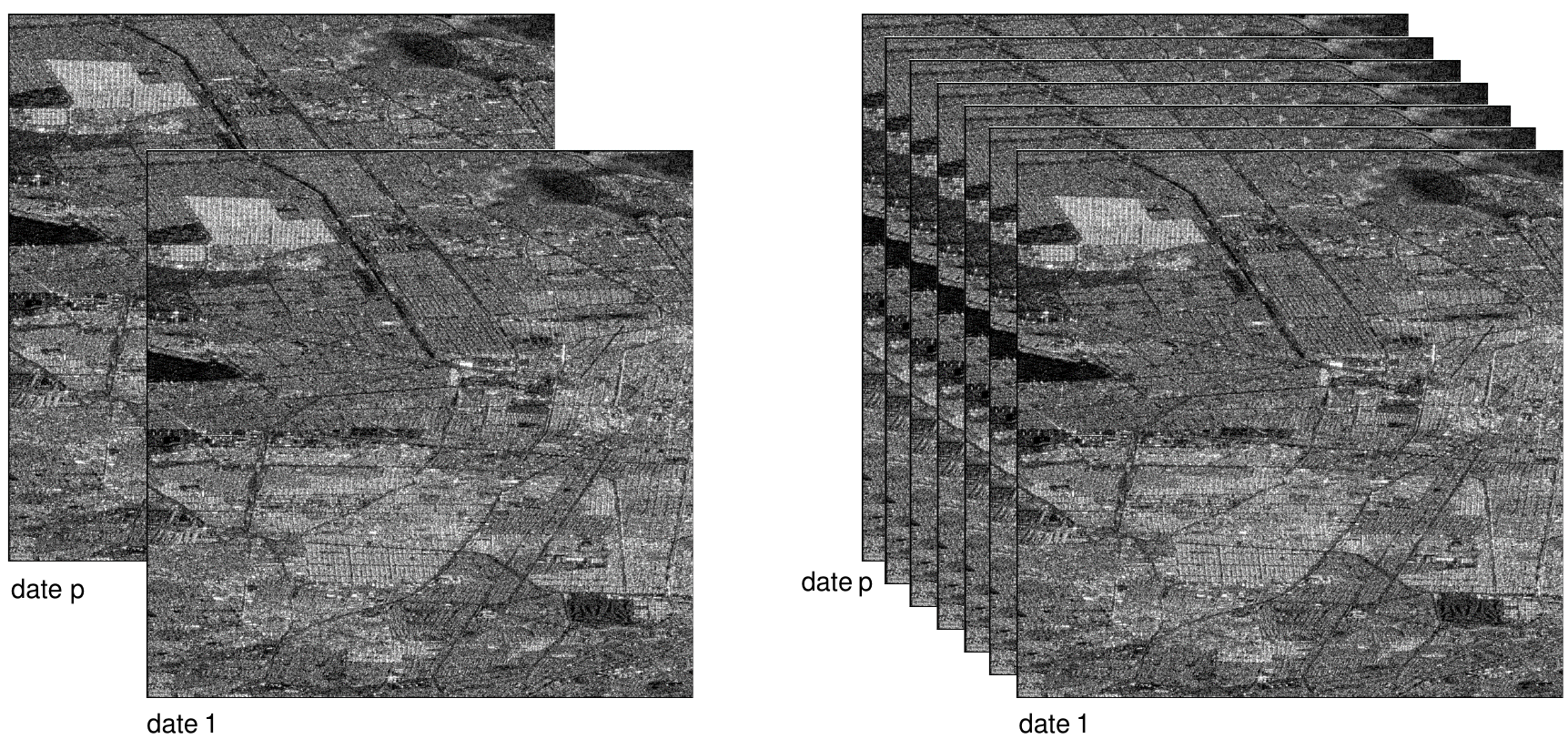}
    \includegraphics[width = 0.46\linewidth]{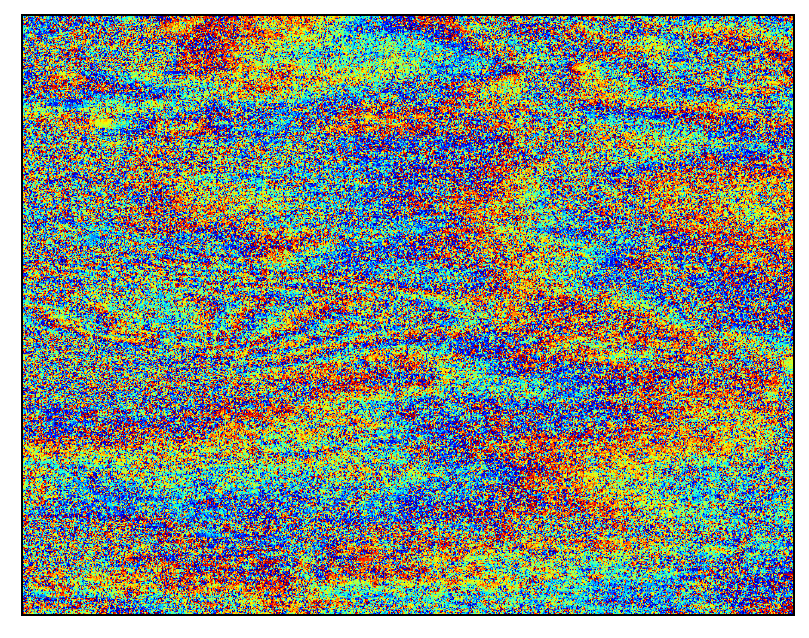}
    \includegraphics[width = 0.46\linewidth]{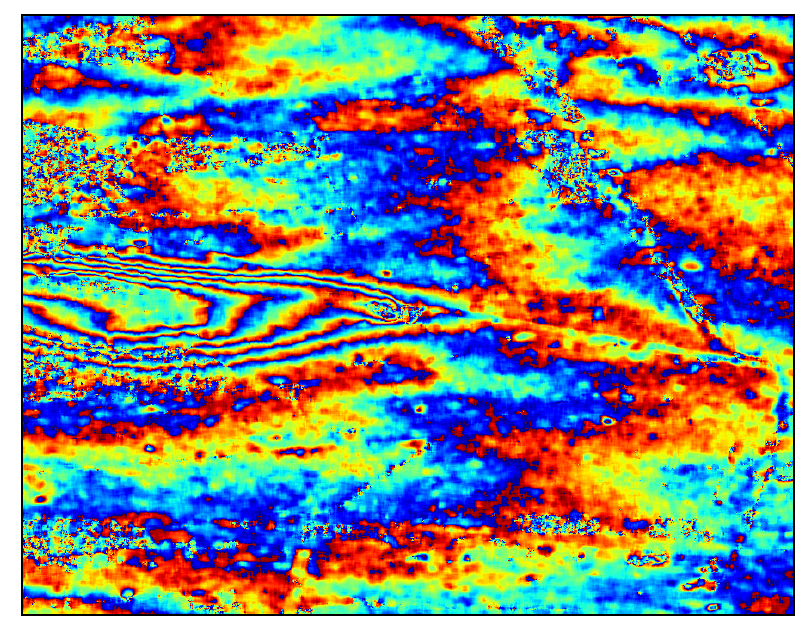} 
    
    \caption{Naive interferogram formed between two SAR images (left) and interferogram formed with interferometric phase linking by leveraging an interleaved time series between those two images (right).
    }
    \label{fig:IPL_principle}
\end{figure}


\section{Background}

\label{sec:background}

\subsection{Notations}

The following convention is adopted: italic indicates a scalar quantity, lower case boldface indicates a vector quantity and upper case boldface a matrix.
For a complex-valued number $x = a+ib=re^{i\theta}\in\mathbb{C}$, 
$\mathfrak{Re}(x)=a$, $\mathfrak{Im}(x)=b$, ${\rm mod}(x)=r$, ${\rm arg}(x)=\theta$ and $\phi_{\mathbb{T}}(x)= e^{i\theta}$.
These operators overload to matrix entries by being applied element-wise.
$\mathcal{H}_p^{++}$ (resp. $\mathcal{S}_p^{++}$) is the set of hermitian (resp. symmetric) positive definite matrices.
$\mathbb{T}_p$ is the torus of phase only complex vectors of dimension $p$ as in \eqref{eq:phaseonlytorus}.
The operators ${\rm tr}(\cdot)$ and $|\cdot|$ return respectively the trace and the determinant of a matrix.
The entrywise complex conjugation is denoted $\cdot^*$, while the transpose (resp. transpose conjugate) operation is denoted $\cdot^\top$ (resp. $\cdot^H$).
The Hadamard product is denoted $\circ$.
The matrix $\mathbf{I}$ denotes the identity matrix of appropriate dimension.
A circular multivariate Gaussian vector of mean $\boldsymbol{\mu}$ and covariance matrix $\mathbf{\Sigma}$ is denoted $\mathbf{x}\sim \mathcal{CN}(\boldsymbol{\mu},\mathbf{\Sigma})$.


\subsection{InSAR covariance matrix structure}

\begin{figure}[!t]
	\begin{center}
    \input{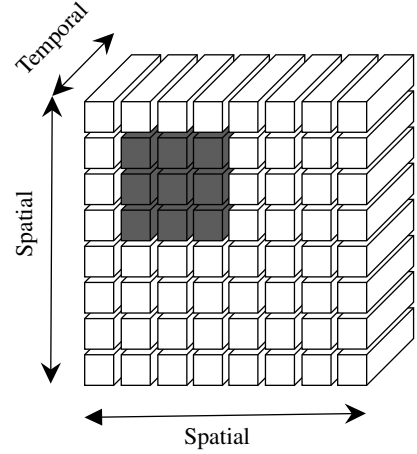}
	\end{center}							
	\caption{
	\label{fig:datacube}
	Stack of $p$ co-registered SAR images and multi-looking window: gray pixels represent the current local patch, denoted $\{\mathbf{x}_i\}_{i=1}^n$.}
\end{figure}

From a given stack of $p$ co-registered SAR images, we consider a sliding window that processes all local patches of $n$ multivariate pixels, as illustrated in Figure \ref{fig:datacube}. 
A pixel patch is denoted as $\{\mathbf{x}_i\}_{i=1}^n$, with $\mathbf{x}_i \in \mathbb{C}^{p},~\forall i\in [\![1,n]\!]$.
A single multivariate pixel $\mathbf{x}_i$ contains a time-series (in chronological order) of the $p$ snapshots, i.e.
\begin{equation}
\mathbf{x}_i = \left[~ x_i^1 ,~\cdots,~x_i^{p}  ~ \right]^{\top} ~\in \mathbb{C}^p.
\end{equation}
The set $\{\mathbf{x}_i\}_{i=1}^n$ is assumed to be a homogeneous patch containing $n$ adjacent pixels with similar scattering and statistical properties. 
From the standard physical considerations about distributed scatterers in SAR, we have that the first and second order moments as follows:
\begin{equation}\label{eq:moments}
\begin{aligned}
& \mathbb{E} \left[ x^q \right] = 0,~\forall~k\in [\![1,p]\!] \\
& \mathbb{E} \left[ x^{q}  ( x^{\ell} )^* \right] = 
\upsilon_{q,\ell}
\sigma_q \sigma_{\ell} e^{j(\theta_q - \theta_{\ell})}
,~\forall (q,\ell) \in [\![1,p]\!]^2 \\
\end{aligned}
\end{equation}
where
\begin{itemize}

\item[$\bullet$] $\sigma_q^2 = \mathbb{E}\left[ x^q(x^q)^H\right] \in \mathbb{R}^+$ is the variance of $x^q$.
The corresponding vector of standard deviations is denoted $\boldsymbol{\sigma} = \left[ \sigma_1 ,~ \cdots,~ \sigma_{p}  \right]$. \vspace{0.2cm}

\item[$\bullet$] $\upsilon_{q,\ell}  \in [0 , 1]$
is the coherence coefficient between $x^q$ and $x^{\ell}$.
The corresponding coherence matrix is denoted $\boldsymbol{\Upsilon}$, with entries $[ \boldsymbol{\Upsilon} ]_{q,\ell} = \upsilon_{q,\ell} \in [0,1]$. Also remark that $[ \boldsymbol{\Upsilon} ]_{\ell,\ell}=1, \forall l\in [\![1,p]\!]$. \vspace{0.2cm}

\item[$\bullet$] $\theta_q$ is the phase instant $q$.
We denote the phase vector $\boldsymbol{\theta} = \left[ \theta_{1},~\cdots,~\theta_{p}\right]$, and the corresponding vector of complex phases\footnote{
Because phase-only complex vectors \cite{smith1999optimum} will be extensively used in this work, we need to distinguish two objects related to the polar decomposition of complex numbers.
For $x = a+ib=re^{i\theta}\in\mathbb{C}$, $\theta$ is indifferently called the phase, angle, or argument.
However, $e^{i\theta}$ will be specifically referred to as the \textit{complex phase}.} is
\begin{equation}\label{eq:w_vector}
\mathbf{w}_{\boldsymbol{\theta}} = 
\left[ e^{j\theta_{1}},~\cdots,~e^{j\theta_{p}}\right] ~\in \mathbb{T}_p,
\end{equation}
where 
\begin{equation} \label{eq:phaseonlytorus}
    \mathbb{T}_p = \{ \mathbf{w}\in\mathbb{C}^p ~|~|[\mathbf{w}]_q| =1, ~\forall q\in [\![1,p]\!] \}
\end{equation}
is the torus of phase-only complex vectors.
By convention, we will use the reference $\theta_1 = 0$, which is equivalent to $[\mathbf{w}_{\boldsymbol{\theta}}]_1=1$.

\end{itemize}
The covariance structure in \eqref{eq:moments} is expressed in matrix form as 
\begin{equation} \label{eq:cov_struct_1}
 \mathbb{E} \left[ \mathbf{x}  \mathbf{x}^H \right]
 = \mathbf{\Sigma}
 = {\rm diag}(\mathbf{w}_{\boldsymbol{\theta}}) {{\boldsymbol{\Psi}}}
  {\rm diag}(\mathbf{w}_{\boldsymbol{\theta}})^H
  = \boldsymbol{\Psi}  \circ (\mathbf{w}_{\boldsymbol{\theta}}
\mathbf{w}_{\boldsymbol{\theta}}^H),
\end{equation}
where $\boldsymbol{\Psi} =    {\boldsymbol{\Upsilon}} \circ\boldsymbol{\sigma} \boldsymbol{\sigma}^\top$ is the coherence matrix scaled by the variance coefficients. 
We can also notice that this decomposition coincides with the modulus-argument decomposition, i.e.:
\begin{equation}\label{eq:mod_arg_decomp}
\mathbf{\Sigma} = {\rm mod}(\mathbf{\Sigma})  \circ \phi_\mathbb{T} (\mathbf{\Sigma}) \overset{\Delta}{=} \boldsymbol{\Psi}  \circ (\mathbf{w}_{\boldsymbol{\theta}}
\mathbf{w}_{\boldsymbol{\theta}}^H).
\end{equation}
where we used the complex phase extraction operator defined as $\phi_\mathbb{T} : x {=} r e^{i\theta} \mapsto e^{i\theta}$ (which extends naturally to matrices by being applied entry-wise). 

From the covariance matrix expression in \eqref{eq:moments} and \eqref{eq:mod_arg_decomp}, we observe that $\mathbf{\Sigma}$
cannot be any covariance matrix in $\mathcal{H}_p^{++}$, as it exhibits the particular phase structure $ {\rm arg}(\mathbf{\Sigma}) = \mathbf{w}_{\boldsymbol{\theta}} \mathbf{w}_{\boldsymbol{\theta}}^H$.
Indeed, if we denote the phase differences between two images indexed $q$ and $\ell$ as $\Delta_{q,\ell} = \theta_q - \theta_{\ell}$, we have
\begin{equation}
\Delta_{q,\ell} +
\Delta_{\ell,j} +
\Delta_{j,q} = 0
\end{equation}
that is satisfied for all triplet $\{q,\ell,j\}$.
Such relationship translates directly into $\mathbf{\Sigma}$ as
\begin{equation} \label{eq:phase_closure}
{\rm arg}  (\mathbf{\Sigma}_{q\ell})   +
{\rm arg}  (\mathbf{\Sigma}_{\ell j})  +
{\rm arg}  (\mathbf{\Sigma}_{j q}) = 0
\end{equation}
because ${\rm arg} (\mathbf{\Sigma}_{q\ell}) = \Delta_{q,\ell}$ from \eqref{eq:moments}.
The aforementioned property is referred to as phase closure, or phase consistency.
It is an important property in MT-InSAR, as it is related to the continuity of physical phenomena, such as Earth displacement.

\subsection{Interferometric phase linking}

In practice, the true covariance matrix $\mathbf{\Sigma}$ of the data is unknown.
The interferometric phases have thus to be estimated solely from the sample set $\{\mathbf{x}_i\}_{i=1}^n$.
From equation \eqref{eq:cov_struct_1}, a naive approach would consists in computing the sample covariance matrix
\begin{equation} \label{eq:SCM}
\mathbf{S} = \frac{1}{n} \sum_{i=1}^n \mathbf{x}_i \mathbf{x}_i^H ,
\end{equation}
and identifying the phase difference from the arguments of its entries, i.e.,
\begin{equation} \label{eq:naive_ipl}
 \hat{\Delta}_{q,\ell} =     {\rm arg}  \ (\mathbf{S}_{q\ell})
\end{equation}
Unfortunately, this simple estimate is relatively inaccurate (especially when $|q-\ell|$ increases, due to the temporal decorrelation).
It also provides a series of phase difference estimates that do not satisfy the phase closure \eqref{eq:phase_closure}. 
More accurate estimation procedures consist rather in directly estimating the vector of complex phases $\mathbf{w}_{\boldsymbol{\theta}}$ (or equivalently, the phase vector $\boldsymbol{\theta}$) from $\mathbf{S}$ by leveraging the prior structure \eqref{eq:mod_arg_decomp}.
The process is referred to as phase triangulation, or interferometric phase linking (IPL), for which numerous algorithms have been developed over the years (cf. \cite{minh2023interferometric} for a recent overview).
In this scope, the next section proposes a general framework capable of encompassing most of the existing methods, providing a modular structure for their extensions.

\section{COFI-PL Framework}

\label{sec:COFIPL_framework}

This section reformulates IPL as a generic covariance fitting problem, whose corresponding framework will be referred to as COFI-PL.
Covariance fitting (or covariance matching) is a widely employed technique in array processing, which consists in refining the structure of an input covariance matrix estimator by minimizing a projection criterion \cite{Ottersten1998, Shahbazpanahi2004, Yardibi2010, Stoica2011, Meriaux2017, Mériaux2019, Mériaux2019b,Mériaux2021}.
In this setup, IPL can be interpreted as projecting (according to some distance or divergence) the input estimate to the set of matrices that satisfy the phase closure property.
In particular, we will see that most of the established MLE-inspired algorithms (e.g., from \cite{Guarnieri2008, Ferretti2011_squee, Cao2015}) appear as a special case of COFI-PL when considering the Kullback-Leibler (KL) divergence as a fitting criterion.

Given any (possibly regularized) plug-in estimate of the covariance matrix $\mathbf{\Sigma}$, denoted $\tilde{\mathbf{\Sigma}}$, that does not satisfy the phase closure, our objective is to determine the ``best'' phase-constrained fitting given some distance criterion linking $\tilde{\mathbf{\Sigma}}$ and its modulus denoted
\begin{equation}
    \tilde{\mathbf{\Psi}} \overset{\Delta}{=} {\rm mod}(\tilde{\mathbf{\Sigma}}).
\end{equation}
The problem is formulated as 
\begin{equation}\label{eq:COFI-PL}
\begin{array}{c l}
\underset{\mathbf{w}_{\boldsymbol{\theta}}}{\rm minimize}
&
d^2 (  \tilde{\mathbf{\Sigma}}
,
\tilde{\boldsymbol{\Psi}}
   \circ
   \mathbf{w}_{\boldsymbol{\theta}}\mathbf{w}_{\boldsymbol{\theta}}^H 
)
\\
{\rm subject~to} 
&
\mathbf{w}_{\boldsymbol{\theta}} \in \mathbb{T}_p
\\
&
\theta_1=0
\end{array}
\end{equation}
where $d$ is a matrix distance (or divergence) that will be specified later on.
Note that the only variable is the vector $\mathbf{w}_{\boldsymbol{\theta}}$, so we will use the compact notation 
\begin{equation}
 f^d_{\tilde{\mathbf{\Sigma}}} ( \mathbf{w}_{\boldsymbol{\theta}}) = 
d^2 (  \tilde{\mathbf{\Sigma}}
,
\tilde{\boldsymbol{\Psi}}
   \circ
   \mathbf{w}_{\boldsymbol{\theta}}\mathbf{w}_{\boldsymbol{\theta}}^H 
)
\end{equation}
for the objective in \eqref{eq:COFI-PL}.
Also remark that the objective function in \eqref{eq:COFI-PL} is invariant to a constant phase-shift of all entries in $\mathbf{w}_{\boldsymbol{\theta}}$, thus the constraint $\theta_1=0$ can be discarded, and achieved \textit{a posteriori} by subtracting $\theta_1$ to all the optimized phases.
The generic COFI-PL problem is finally expressed as
\begin{equation}\label{eq:COFI-PL2}
\begin{array}{c l}
\underset{\mathbf{w}_{\boldsymbol{\theta}}}{\rm minimize}
&
f^d_{\tilde{\mathbf{\Sigma}}} ( \mathbf{w}_{\boldsymbol{\theta}}) 
\\
{\rm subject~to} 
&
\mathbf{w}_{\boldsymbol{\theta}} \in \mathbb{T}_p.
\end{array}
\end{equation}
This formulation then offers a multitude of options concerning:
\begin{itemize}
    \item The construction of the covariance matrix plug-in $\hat{\mathbf{\Sigma}}$
    and its possible regularization $\tilde{\mathbf{\Sigma}}$.
    
    \item The choice of the matrix distance $d$ according to a geometry of interest.

    \item The optimization method to address the constraint $\mathbf{w}_{\boldsymbol{\theta}} \in \mathbb{T}_p$ efficiently.
\end{itemize}
With these options specified, the corresponding instance of COFI-PL then defines a complete processing chain from the data $\{\mathbf{x}_i\}_{i=1}^n$ to the phase estimates, as illustrated in Figure \ref{fig:COFI-PLchain}.
The following sections \ref{sec:cm_est}, \ref{sec:cm_regul}, and \ref{sec:matrix_dist} will provide an overview of practical design options suited to SAR data.
Then sections \ref{sec:mm} and \ref{sec:riem_opt} will present two optimization frameworks that can be used to solve \eqref{eq:COFI-PL2} under the constraint $\mathbf{w}_{\boldsymbol{\theta}} \in \mathbb{T}_p$.

\begin{figure*}[!t]
    \tikzset{every picture/.style={line width=0.75pt}} 

\begin{tikzpicture}[x=0.75pt,y=0.75pt,yscale=-1,xscale=1]

\draw    (1.79,50.25) -- (42.33,50.25) ;
\draw [shift={(44.33,50.25)}, rotate = 180] [color={rgb, 255:red, 0; green, 0; blue, 0 }  ][line width=0.75]    (10.93,-3.29) .. controls (6.95,-1.4) and (3.31,-0.3) .. (0,0) .. controls (3.31,0.3) and (6.95,1.4) .. (10.93,3.29)   ;
\draw   (45.08,26.85) .. controls (45.08,18.1) and (52.18,11) .. (60.93,11) -- (136.94,11) .. controls (145.7,11) and (152.79,18.1) .. (152.79,26.85) -- (152.79,74.4) .. controls (152.79,83.15) and (145.7,90.25) .. (136.94,90.25) -- (60.93,90.25) .. controls (52.18,90.25) and (45.08,83.15) .. (45.08,74.4) -- cycle ;

\draw    (152.79,50.25) -- (193.33,50.25) ;
\draw [shift={(195.33,50.25)}, rotate = 180] [color={rgb, 255:red, 0; green, 0; blue, 0 }  ][line width=0.75]    (10.93,-3.29) .. controls (6.95,-1.4) and (3.31,-0.3) .. (0,0) .. controls (3.31,0.3) and (6.95,1.4) .. (10.93,3.29)   ;
\draw   (196.08,26.85) .. controls (196.08,18.1) and (203.18,11) .. (211.93,11) -- (287.94,11) .. controls (296.7,11) and (303.79,18.1) .. (303.79,26.85) -- (303.79,74.4) .. controls (303.79,83.15) and (296.7,90.25) .. (287.94,90.25) -- (211.93,90.25) .. controls (203.18,90.25) and (196.08,83.15) .. (196.08,74.4) -- cycle ;

\draw    (303.79,50.25) -- (344.33,50.25) ;
\draw [shift={(346.33,50.25)}, rotate = 180] [color={rgb, 255:red, 0; green, 0; blue, 0 }  ][line width=0.75]    (10.93,-3.29) .. controls (6.95,-1.4) and (3.31,-0.3) .. (0,0) .. controls (3.31,0.3) and (6.95,1.4) .. (10.93,3.29)   ;
\draw   (347.08,26.85) .. controls (347.08,18.1) and (354.18,11) .. (362.93,11) -- (438.94,11) .. controls (447.7,11) and (454.79,18.1) .. (454.79,26.85) -- (454.79,74.4) .. controls (454.79,83.15) and (447.7,90.25) .. (438.94,90.25) -- (362.93,90.25) .. controls (354.18,90.25) and (347.08,83.15) .. (347.08,74.4) -- cycle ;

\draw    (454.79,50.25) -- (495.33,50.25) ;
\draw [shift={(497.33,50.25)}, rotate = 180] [color={rgb, 255:red, 0; green, 0; blue, 0 }  ][line width=0.75]    (10.93,-3.29) .. controls (6.95,-1.4) and (3.31,-0.3) .. (0,0) .. controls (3.31,0.3) and (6.95,1.4) .. (10.93,3.29)   ;
\draw   (498.08,26.85) .. controls (498.08,18.1) and (505.18,11) .. (513.93,11) -- (589.94,11) .. controls (598.7,11) and (605.79,18.1) .. (605.79,26.85) -- (605.79,74.4) .. controls (605.79,83.15) and (598.7,90.25) .. (589.94,90.25) -- (513.93,90.25) .. controls (505.18,90.25) and (498.08,83.15) .. (498.08,74.4) -- cycle ;
\draw    (605.79,50.25) -- (646.33,50.25) ;
\draw [shift={(648.33,50.25)}, rotate = 180] [color={rgb, 255:red, 0; green, 0; blue, 0 }  ][line width=0.75]    (10.93,-3.29) .. controls (6.95,-1.4) and (3.31,-0.3) .. (0,0) .. controls (3.31,0.3) and (6.95,1.4) .. (10.93,3.29)   ;

\draw (-10,28) node [anchor=north west][inner sep=0.75pt]   [align=left] {$\{\mathbf{x}_i\}_{i=1}^n$};
\draw (165,28) node [anchor=north west][inner sep=0.75pt]   [align=left] {$\hat{\mathbf{\Sigma}}$};
\draw (315,28) node [anchor=north west][inner sep=0.75pt]   [align=left] {$\tilde{\mathbf{\Sigma}}$};
\draw (462,28) node [anchor=north west][inner sep=0.75pt]   [align=left] {\eqref{eq:COFI-PL2}};
\draw (620,28) node [anchor=north west][inner sep=0.75pt]   [align=left] {$\hat{\mathbf{w}}_{\boldsymbol{\theta}}$};

\draw (59,43) node [anchor=north west][inner sep=0.75pt]   [align=left] {CM estimation};
\draw (199,43) node [anchor=north west][inner sep=0.75pt]   [align=left] {CM regularization};
\draw (351,43) node [anchor=north west][inner sep=0.75pt]   [align=left] {Matrix distance $d$};
\draw (513,43) node [anchor=north west][inner sep=0.75pt]   [align=left] {Optim. on $\mathbb{T}_p$};

\draw (45.08,94) node [anchor=north west][inner sep=0.79pt]   [align=left] {\scriptsize module options in Sec. \ref{sec:cm_est}};
\draw (196.08,94) node [anchor=north west][inner sep=0.79pt]   [align=left] {\scriptsize module options in Sec. \ref{sec:cm_regul}};
\draw (347.08,94) node [anchor=north west][inner sep=0.79pt]   [align=left] {\scriptsize module options in Sec. \ref{sec:matrix_dist}};
\draw (498.08,94) node [anchor=north west][inner sep=0.79pt]   [align=left] {\scriptsize module options in Sec. \ref{sec:mm}-\ref{sec:riem_opt}};

\end{tikzpicture}
    \caption{A modular COFI-PL chain of process: the covariance matrix plug-in $\tilde{\mathbf{\Sigma}}$
    is estimated from the samples $\{\mathbf{x}_i\}_{i=1}^n$ with a possible regularization. 
    The choice of a matrix distance $d$ then specifies the covariance fitting optimization problem \eqref{eq:COFI-PL2}.
    Interferometric phase estimates are obtained by solving this problem using optimization methods on the $n$-torus of phase-only vectors $\mathbb{T}_p$.    
    }
    \label{fig:COFI-PLchain}
\end{figure*}
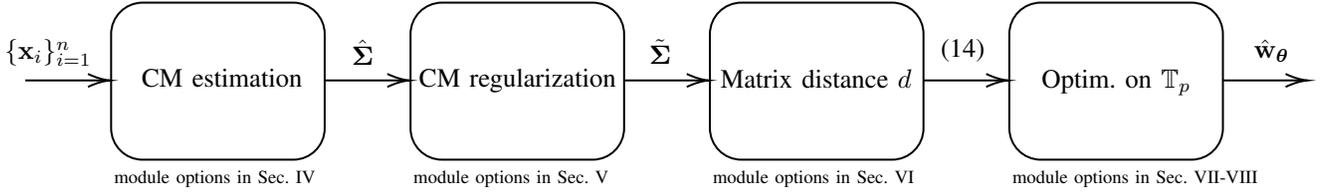

\section{Unstructured covariance matrix estimation}
\label{sec:cm_est}

The first module requires to construct an estimate of the covariance matrix $\hat{\mathbf{\Sigma}}$ from the sample set $\{\mathbf{x}_i\}_{i=1}^n$.
In this setup, it is interesting to link the chosen estimation method to underlying assumptions on the statistical model of the data \cite{deng2017statistical}.
Notably, SAR data can be non-Gaussian \cite{mian2018new, Vu2023robust}, which motivates the use of robust estimation methods.
This section overviews relevant plug-in estimates for SAR data, and their corresponding model assumptions.

\subsection{Sample covariance matrix}

A common assumption in SAR is to consider that each sample is independent and identically distributed according to a centered complex circular Gaussian model, denoted $\mathbf{x}\sim\mathcal{CN}(\mathbf{0},\mathbf{\Sigma})$.
This model is particularly relevant for low resolution SAR images, i.e., when we can assume that each pixel gathers the sum of the contributions from many scatterers, and apply the central limit theorem.
In this case, the pixel patch has the negative log-likelihood
\begin{equation} \label{eq:gauss_likelihood}
    \mathcal{L}_{\mathcal{N}}(\{\mathbf{x}_i\}_{i=1}^n|\mathbf{\Sigma})  \propto {\rm tr} (\mathbf{S}\mathbf{\Sigma}^{-1} ) + \log | \mathbf{\Sigma} |
\end{equation}
with $\mathbf{S}$ as in \eqref{eq:SCM}.
In this case, the maximum likelihood estimator of $\mathbf{\Sigma}$ is the sample covariance matrix $\mathbf{S}$.
This estimate is the most widely used plug-in in IPL \cite{Ferretti2011_squee, Cao2015}.

\subsection{$M$-estimators of the scatter}

At high resolution, the Gaussian assumption is often no longer valid, and empirical histograms tend to exhibit heavy-tails (see, e.g., \cite{mian2018new, Vu2023robust}).
The centered circular complex elliptically symmetric (CES) distributions provides a better fit for such data.
CES models, denoted $\mathbf{x}\sim\mathcal{CES}(\mathbf{0},\mathbf{\Sigma},g)$, correspond to the following negative log-likelihood function:
\begin{equation} \label{eq:ces_likelihood}
\mathcal{L}_{\mathcal{E}}^g
(\{\mathbf{x}_i\}_{i=1}^n|\mathbf{\Sigma})
\propto
\frac{1}{n}\sum_{i=1} \rho (   \mathbf{x}_i^H \mathbf{\Sigma}^{-1} \mathbf{x}_i    )
+ \log | \mathbf{\Sigma} |
\end{equation}
with $\rho(t) = - \log g(t)$, and where $g$ is referred to as the density generator.
The special case $g(t)=e^{-t}$ yields back the Gaussian model, while many other options allow for modelling other heavy-tailed distributions \cite{ollila2012complex}.
An $M$-estimator of the scatter $ \mathbf{\Sigma}_{M}$ is then defined as the solution of the following fixed-point equation 
\begin{equation}
\label{eq:mest}
 \mathbf{\Sigma}_{M} = \frac{1}{n}\sum_{i=1}^{n} u(\mathbf{x}_i^H  \mathbf{\Sigma}_{M}^{-1}\mathbf{x}_i)\mathbf{x}_i\mathbf{x}_i^H
\end{equation}
where $u$ is a real-valued weight function on $[0,\infty)$. 
Such solution exists and is unique under conditions on $u$ and the sample support (notably, $n>p$) discussed in \cite{maronna1976robust, ollila2012complex}.
When $u(t) = - g'(t) / g (t)$,~\eqref{eq:mest} is the maximum likelihood estimator of $\mathbf{\Sigma}$ for $\mathbf{x} \sim \mathcal{CES}(\mathbf{0},\mathbf{\Sigma}, g)$.
Otherwise, $M$-estimators are robust to a mismatch, and we can use a function $u$ that is not necessarily linked to $g$.
Most notably, Tyler's estimator \cite{Tyler1987}, based on the function $u_{T}(t) = p/t$, is distribution free, in the sense that its distribution does not depend on the underlying density generator $g$.
In practice, $M$-estimators still offer better estimation performance compared to the sample covariance matrix as long as the function $u$ assesses for potentially heavy-tailed distribution.
Specifically for IPL, these estimators have been considered as plug-ins in \cite{schmitt2014adaptive, wang2015robust}.

\subsection{Correlation-based plug-in}

First, recall that ${\rm mod} (\mathbf{\Sigma})=\boldsymbol{\Psi} =  {\boldsymbol{\Upsilon}} \circ\boldsymbol{\sigma} \boldsymbol{\sigma}^\top$, where ${\boldsymbol{\Upsilon}}$ is the coherence matrix and $\boldsymbol{\sigma}$ is the vector of standard deviations.
Remark that the scaling by $\boldsymbol{\sigma}$ does not impact the phase structure of $\mathbf{\Sigma}$ in \eqref{eq:cov_struct_1}.
This can motivate the use of a correlation (rather than covariance) matrix estimate $\hat{\mathbf{C}}$ as plug-in, in order to mitigate issues related to the amplitude fluctuation in SAR images (unbalanced backscattered power among all the images).
Such matrix is built from any covariance matrix estimate $\hat{\mathbf{\Sigma}}$ as
\begin{equation} \label{eq:C_standardization}
    \hat{\mathbf{C}} = {\rm diag}(\hat{\mathbf{\Sigma}})^{-1/2} \hat{\mathbf{\Sigma}}\: {\rm diag}(\hat{\mathbf{\Sigma}})^{-1/2}.
\end{equation}
Most notably, the sample correlation matrix, i.e., \eqref{eq:C_standardization} built with $\hat{\mathbf{\Sigma}}=\mathbf{S}$, was successfully leveraged in \cite{Cao2015}, and other works that implicitly standardize the data before computing $\mathbf{S}$.
Still, this estimator can be sensitive to non-Gaussian distributions.
A solution is to turn to robust estimators of the correlation matrix, for which many options exist \cite{shevlyakov2011robust, shevlyakov2016robust}.
We will focus on a simple and highly robust one, that is obtained from the phase-only sample correlation matrix
\begin{equation}
    \mathbf{T} = \frac{1}{n} \sum_{i=1}^n \mathbf{y}_i \mathbf{y}_i^H,
    \label{eq:phase_only_SCM}
\end{equation}
with $\mathbf{y}=\phi (\mathbf{x})$.
As this estimator projects each entry of the samples to the unit sphere, it is inherently robust to any underlying distributions of the modulus of the marginals $x^q,~\forall q\in[\![1,p]\!]$ (cf. \cite{woodbridge2017signal} for an example of such distributions).


\section{Covariance matrix regularization}
\label{sec:cm_regul}

The second module consists in applying a regularization process to the estimator provided by the first one.
Indeed, the sample support $n$ is limited by the size of the sliding window (that sets the spatial resolution).
For long time series, this can lead to situations where $n\simeq p$ or $n<p$, in which the plain estimate $\hat{\mathbf{\Sigma}}$ is inaccurate.
Applying some form of regularization to the estimate in this case greatly improves the accuracy of phase estimation.
This section thus overviews how to construct a regularized estimator $\tilde{\mathbf{\Sigma}}$ from a plug-in estimate $\hat{\mathbf{\Sigma}}$ with different approaches motivated by IPL, and assumptions on SAR data.
These come at the cost of regularization parameter selection, so we also discuss references that address this issue.

\subsection{Shrinkage to identity}

At low sample support, covariance matrix estimates are usually ill conditioned ($n\simeq p$), or even not invertible ($n<p$).
This poses a major issue in IPL because most of the fitting distances are constructed from the inverse of the plug-in estimate (cf. Section  \ref{sec:matrix_dist}).
A practical solution is to operate a shrinkage of the estimate to a scaled identity matrix, in order to leverage some form of bias-variance trade-off.
A popular formulation of such regularization preserves the scale (i.e., trace) of the estimate, and is defined for any plug-in estimate $\hat{\mathbf{\Sigma}}$ as
\begin{equation} \label{eq:LWshrink}
    \tilde{\mathbf{\Sigma}}(\beta)
    = 
    \beta \hat{\mathbf{\Sigma}}
    +
    (1-\beta) \frac{{\rm tr}(\hat{\mathbf{\Sigma}}) }{p} \mathbf{I}
\end{equation}
with $\beta\in [0,1]$.
This regularization is also referred to as spectral shrinkage, as it shrinks eigenvalues of the plug-in estimate towards their mean.
When the plug-in estimate is the sample covariance matrix $\mathbf{S}$, the adaptive selection of the regularization parameter $\beta$ for minimizing the mean squared error has been studied in \cite{Ledoit2004} (assuming finite $4^{\text{th}}$ order moments), \cite{Chen2011} (assuming Gaussian data), and \cite{ollila2019optimal} (assuming elliptically distributed data).
Extension to the $M$-estimators as in \eqref{eq:mest} has been studied in \cite{ollila2020shrinking}.

\subsection{Low-rank approximation}

The empirical spectrum of multivariate SAR pixel patches often exhibits a low-rank structure (see e.g. experiments in \cite{mian2020robust, Vu2023robust}).
This means that most of the data variance lies in a rank-$k$ linear subspace, which justifies the use of principal component analysis\footnote{
In practice the rank $k$ is usually set fixed for the whole image, i.e. for processing all pixel patches indifferently.
Though it comes at a heavy computational cost, this process could be refined by using adaptive rank estimation, e.g., using model order selection methods \citep{stoica2004model}.
} \cite{jolliffe2003principal}.
The corresponding decomposition of the covariance matrix has often been 
used in IPL to improve the quality of the plug-in estimate \cite{Cao2015}.
Let the eigenvalue decomposition of this plug-in estimate be denoted $\hat{\mathbf{\Sigma}} \overset{\rm EVD}{=} \sum_{r=1}^p \hat{\lambda}_r \hat{\mathbf{u}}_r \hat{\mathbf{u}}_r^H$.
The shrinkage approaches that leverage a low-rank structure within this decomposition consider either the rank-$k$ approximation of the plug-in estimate 
\begin{equation} \label{eq:low_rank_approx}
\tilde{\mathcal{P}}_k (\hat{\mathbf{\Sigma}} ) = \sum_{r=1}^k \hat{\lambda}_r \hat{\mathbf{u}}_r \hat{\mathbf{u}}_r^H ,
\end{equation}
or its projection on the set of rank-$k$ plus scaled identity, defined as
\begin{equation}  \label{eq:low_rank_approx2}
{\mathcal{P}}_k (\hat{\mathbf{\Sigma}} ) =
\sum_{r=1}^k \hat{\lambda}_r \hat{\mathbf{u}}_r \hat{\mathbf{u}}_r^H 
+
\sum_{r=k+1}^p \bar{\lambda}_r \hat{\mathbf{u}}_r \hat{\mathbf{u}}_r^H ,
\end{equation}
with $\bar{\lambda}_r = \sum_{r=k+1}^p \hat{\lambda}_r /(p-k)$.
This second option is advocated in this work for two main reasons.
The first is theoretical, as applying $\mathcal{P}_k$ to the sample covariance matrix $\mathbf{S}$ yields the maximum likelihood estimator of a low-rank Gaussian signal plus white Gaussian noise \cite{Tipping1999, kang2014rank}.
The second is practical: the operator $\mathcal{P}_k$ gives an invertible matrix and improves the conditioning of the plug-in estimate by uplifting its lowest eigenvalue to the average of the $p-k$ lowest ones.
Similar to the shrinkage to identity, this property is instrumental to compute distances $d$ whose expressions involve matrix inverses.

\subsection{Covariance matrix tapering}

SAR image stacks suffer from targets decorrelation over time, which is why IPL exploits the redundancy brought by all image pairs in order to construct the interferograms.
However, time frames that exhibit a prohibitively low coherence should intuitively be excluded from this construction, which motivated the development of methods to determine which pairs of images are exploited or disregarded \cite{ansari2020study, shen2023adppl}.
In this scope, a simple idea consists in processing only pairs contained in a sliding temporal window of bandwidth $b$ (where the coherence is assumed to remain high enough).
Within the IPL framework, this translates into forcing a banded structure in the covariance matrix plug-in.
Let $\mathbf{W}(b)$ be a banding-type\footnote{Note that other tapering templates
could be envisioned \cite{cai2010optimal, guerci1999theory}, but we focus only on the most relevant to the considered applications.} tapering matrix with bandwidth $b$ \cite{bickel2008regularized, bickel2008covariance}
\begin{equation}
    [\mathbf{W}(b)]_{ij} = 
    \left\{
    \begin{array}{cl}
        1  & \text{if}~|i-j|\leq b\\
        0  & \text{otherwise}.
    \end{array}
    \right.
\end{equation}
Covariance matrix tapering, also referred to as Hadamard regularization, involves producing the regularized estimate
\begin{equation} \label{eq:tapering}
    \tilde{\mathbf{\Sigma}} (b) = \mathbf{W}(b) \circ \hat{\mathbf{\Sigma}},
\end{equation}
which then naturally exhibits a banded structure.
Adaptive procedures for the bandwidth selection can be found in \cite{bickel2008regularized, bickel2008covariance, ollila2022regularized}.
The joint use of shrinkage to identity and tapering was advocated for IPL in \cite{zwieback2022cheap}.
The optimal (in the sense of mean-squared error) adaptive selection of the parameters $(\beta,b)$ for this regularization was studied in \cite{ollila2022regularized} and applied to IPL in \cite{bai2023lamie}.

\subsection{Joint estimation and regularization}

\label{sec:estimreguljoint}

For completeness, we mention that some covariance matrix plug-ins can be constructed by ``merging'' the estimation and regularization modules.
A main example is that $M$-estimators do not exist for $n<p$, so the regularization cannot be performed afterwards, and has to be included within the robust estimation process.
Regularized $M$-estimators, i.e., minimizers of $\mathcal{L}_{\mathcal{E}}^g
(\{\mathbf{x}_i\}_{i=1}^n|\mathbf{\Sigma})$ in \eqref{eq:ces_likelihood}
plus an additive penalty term, were studied in \cite{wiesel2011unified, sun2014regularized, pascal2014generalized, ollila2014regularized}, and their use for IPL was discussed in \cite{even2018insar}.
Low-rank structured $M$-estimators, i.e., minimizers of $\mathcal{L}_{\mathcal{E}}^g
(\{\mathbf{x}_i\}_{i=1}^n|\mathbf{\Sigma})$ in \eqref{eq:ces_likelihood} under the low-rank structure constraint as in \cite{Tipping1999, kang2014rank}, were investigated \cite{sun2016robust, mian2020robust}.

\section{Matrix distances}
\label{sec:matrix_dist}

The previous sections presented two modules dedicated to the construction of a (possibly regularized) plug-in estimate of the covariance matrix $\tilde{\mathbf{\Sigma}}$.
From this estimate, the third module formulates the optimization problem \eqref{eq:COFI-PL2}, which established a covariance matching type IPL.
This process simply requires to select a matrix distance, for which many options are available.
This section overviews prominent ones, and explores their connections to statistical assumptions and existing IPL methods.

\subsection{Kullback-Leibler divergence and maximum likelihood estimation approaches}

The Kullback-Leibler (KL) divergence measures the dissimilarity between two probability density functions.
Its expression between two centered Gaussian distributions $\mathbf{x} \sim \mathcal{CN}(0, \mathbf{\Sigma}_1)$ and $\mathbf{x}\sim \mathcal{CN}(0, \mathbf{\Sigma}_2)$ is:
\begin{multline}
    \label{eq:KL_gauss}
    {\rm KL}(\mathcal{CN}(0, \mathbf{\Sigma}_1 )~||~\mathcal{CN}(0, \mathbf{\Sigma}_2  ))
\\ = {\rm tr}(\mathbf{\Sigma}_2^{-1}\mathbf{\Sigma}_1) + \log |\mathbf{\Sigma}_2\mathbf{\Sigma}_1^{-1}|-p,
\end{multline}
which provides a matrix divergence between $\mathbf{\Sigma}_1$ and $\mathbf{\Sigma}_2$.
Setting $\mathbf{\Sigma}_1 = \tilde{\mathbf{\Sigma}} $, $\mathbf{\Sigma}_2 = \tilde{\boldsymbol{\Psi}}
\circ \mathbf{w}_{\boldsymbol{\theta}}  \mathbf{w}_{\boldsymbol{\theta}}^H $, keeping $\mathbf{w}_{\boldsymbol{\theta}}$ as the only variable, and simplifying the expression, we obtain the KL-IPL objective function as
\begin{equation} \label{eq:kl_COFI-PL}
f^{\rm KL}_{\tilde{\mathbf{\Sigma}}}( \mathbf{w}_{\boldsymbol{\theta}})  = \mathbf{w}_{\boldsymbol{\theta}}^H
( \tilde{\boldsymbol{\Psi}}^{-1} \circ \tilde{\mathbf{\Sigma}} ) 
\mathbf{w}_{\boldsymbol{\theta}}.
\end{equation}
This objective function is the most widely employed for IPL, especially when choosing the sample covariance matrix as plug-in estimate.
However, the expression \eqref{eq:kl_COFI-PL} with $\tilde{\mathbf{\Sigma}} = \mathbf{S}$ is usually obtained from the perspective of maximum likelihood estimation under the Gaussian model when assuming that ${\rm mod}(\mathbf{\Sigma})$ (coherence matrix scaled by the variance coefficients) is known \cite{Guarnieri2008, Ferretti2011_squee, Cao2015}. 
Owing to similar expressions of the Gaussian log-likelihood in \eqref{eq:gauss_likelihood} and the KL divergence \eqref{eq:KL_gauss} when using $\mathbf{S}$ as the plug-in estimate, the two approaches fall back on the same objective in this case.
There is a slight difference between these approaches when using regularization, which is discussed further in the overview Section \ref{sec:sota}.
A last remark is that the KL divergence is not symmetric, so it is not a proper matrix distance.
Its symmetric counterpart (i.e., inverting the roles of $\mathbf{\Sigma}_1$ and $\mathbf{\Sigma}_2$), and its symmetrized version could also be envisioned as options. 
Still, we focus only on the chosen formulation because of its direct link to the Gaussian maximum likelihood estimator, and well-established IPL algorithms.

\subsection{Frobenius norm and least-squares estimator}

The euclidean distance, also referred to as ``flat metric'', between two symmetric matrices is defined as
\begin{equation}
    d_{E}^2 ( \mathbf{\Sigma}_1, \mathbf{\Sigma}_2) = || \mathbf{\Sigma}_1 - \mathbf{\Sigma}_2 ||^2_F .
\end{equation}
By setting $\mathbf{\Sigma}_1 = \tilde{\mathbf{\Sigma}} $, $\mathbf{\Sigma}_2 = \tilde{\boldsymbol{\Psi}}
\circ \mathbf{w}_{\boldsymbol{\theta}}  \mathbf{w}_{\boldsymbol{\theta}}^H $, while keeping $\mathbf{w}_{\boldsymbol{\theta}}$ as the only variable, and simplifying the expression, we obtain the LS-IPL (for least squares) objective function as
\begin{equation} \label{eq:ls_COFI-PL}
    f^{\rm LS}_{\tilde{\mathbf{\Sigma}}} (\mathbf{w}_{\boldsymbol{\theta}} )
    =
    -2 \mathbf{w}_{\boldsymbol{\theta}}^H
    ( \tilde{\boldsymbol{\Psi}} \circ \tilde{\mathbf{\Sigma}} ) 
    \mathbf{w}_{\boldsymbol{\theta}}
    + {\rm const.},
\end{equation}
This objective function appears as similar to the KL-IPL problem \eqref{eq:kl_COFI-PL}, i.e., a quadratic form on $\mathbb{T}_p$.
However, KL-IPL is formulated as a minimization of the quadratic form involving $ \tilde{\boldsymbol{\Psi}}^{-1} \circ \tilde{\mathbf{\Sigma}} $, while LS-IPL is formulated as a maximization (because of the minus sign in \eqref{eq:ls_COFI-PL}) one involving $ \tilde{\boldsymbol{\Psi}} \circ \tilde{\mathbf{\Sigma}} $.
Though the LS-IPL formulation has driven more research due to its maximum likelihood grounds \cite{minh2023interferometric}, recent studies \cite{vu2023covariance, bai2023lamie} evidenced the practical use of $f^{\rm LS}$.
A notable interest of this objective function is that no covariance matrix inversion is required: besides the reduction of the computational complexity, it can also mitigate inaccuracies arising from the poor conditioning of the plug-in estimate (without requiring any regularization process).

\subsection{Weighted Frobenius norm and EXIP criterion}

The weighted Frobenius norm between two symmetric matrices is defined as:
\begin{equation}
    d_{\rm WF}^2 
    ({\mathbf{\Sigma}}_1,{\mathbf{\Sigma}}_2)
    =
    || 
    \mathbf{\mathbf{H}}^{-1/2}
    (
    {\mathbf{\Sigma}}_1
    -
    {\mathbf{\Sigma}}_2 
    )
    \mathbf{\mathbf{H}}^{-1/2}
    ||_F^2,
\end{equation}
where $\mathbf{H}$ is a whitening-type weight matrix to be fixed.
Among possible choices, we will focus on $\mathbf{H}=\tilde{\mathbf{\Sigma}}$, 
that yields the extended invariance principle (EXIP) and COMET-type estimators, which hold interesting statistical properties \cite{Ottersten1998, Shahbazpanahi2004, Yardibi2010, Stoica2011, Meriaux2017, Mériaux2019, Mériaux2019b, Mériaux2021}.
Such EXIP approach has been used in InSAR in \cite{yunjun2019small, guarnieri2007hybrid, hu2023fim} (though not always directly related to IPL formulations).
For this choice, setting $\mathbf{\Sigma}_1 = \tilde{\mathbf{\Sigma}} $ and $\mathbf{\Sigma}_2 = \tilde{\boldsymbol{\Psi}}
\circ \mathbf{w}_{\boldsymbol{\theta}}  \mathbf{w}_{\boldsymbol{\theta}}^H $ gives the following WLS-IPL (for weighted least squares) objective function
\begin{equation}
\begin{aligned}  
    f^{\rm WLS}_{\tilde{\mathbf{\Sigma}}}
    & =
    || 
    \mathbf{I}
    -
    \tilde{\mathbf{\Sigma}}^{-1/2}
    (\tilde{\boldsymbol{\Psi}}
    \circ \mathbf{w}_{\boldsymbol{\theta}} \mathbf{w}_{\boldsymbol{\theta}}^H) 
    \tilde{\mathbf{\Sigma}}^{-1/2}
    ||_F^2,
\end{aligned}
\end{equation}
which unfortunately, cannot be simplified into a simpler quadratic form as KL-IPL and LS-IPL.


\section{Majorization-minimization on $\mathbb{T}_p$}
\label{sec:mm}

The options discussed in the previous section allow us to construct optimization problems as in \eqref{eq:COFI-PL2} in order to perform IPL.
This section presents majorization-minimization on $\mathbb{T}_p$ in order to address the resolution of these problems.
This optimization framework has been successfully leveraged for KL-IPL \cite{Vu2023robust, Vu2022igarss} and LS-IPL \cite{vu2023covariance}.
Though it is not applicable to all objective functions, its major practical interest lies in the fact that it leads to simple and scalable algorithms.

\subsection{Majorization-minimization}

Majorization-minimization is briefly reviewed here with notations that match the problem \eqref{eq:COFI-PL2} for convenience.
More details on this framework can be found in \cite{hunter2004tutorial, sun2016majorization}. 
We consider an optimization problem of the form 
\begin{equation}
\begin{array}{c l}\label{eq:problem_phases2}
\underset{{\mathbf{w}}\in\mathbb{T}_p}{\rm minimize}
& 
f(\mathbf{w}) .
\end{array}
\end{equation}
The majorization-minimization algorithm is an iterative optimization procedure that operates with two steps:
\begin{enumerate}
    \item \textit{Majorization}: at current point $\mathbf{w}_t$, find a surrogate function $g(\cdot |\mathbf{w}_t)$ so that it is tangent to the objective, $f(\mathbf{w}_t) = g(\mathbf{w}_t | \mathbf{w}_t)$, and majorizes it, i.e.,
    \begin{equation}
        f(\mathbf{w})\leq  g(\mathbf{w} | \mathbf{w}_t),~ \forall \mathbf{w} \in \mathbb{T}_p
    \end{equation}

    \item \textit{Minimization}: obtain the next iterate as 
    \begin{equation}
        \mathbf{w}_{t+1} = {\rm argmin}_{\mathbf{w}\in \mathbb{T}_p } ~ g(\mathbf{w} | \mathbf{w}_t).
    \end{equation}
\end{enumerate}
This algorithm enjoys nice convergence properties \cite{razaviyayn2013unified} (being constrained to the compact set $\mathbb{T}_p$ can be accounted for with the same arguments as in \cite{soltanalian2014designing, breloy2021majorization}).
It notably ensures a monotonic decrement of the objective function at each step.
The main interest of this approach is that it can yield a sequence of sub-problems that are easily solved if the surrogate functions are well constructed.

\subsection{Lemmas for quadratic forms of phase only complex vectors}

This section presents useful surrogates functions and their closed form minimizer for quadratic forms on $\mathbb{T}_p$.
Let $\mathbf{H}\succcurlyeq \mathbf{0}$ be a hermitian positive semi-definite matrix, we have the following lemmas.

\begin{lemma} \label{lemma:ccvqf}
    The concave quadratic form $f:\mathbf{w} \mapsto - \mathbf{w}^H \mathbf{H} \mathbf{w}$ is majorized at point $\mathbf{w}_t$ by the surrogate function
    \begin{equation}
        g ( \mathbf{w} | \mathbf{w}_t )  = 
        - 2 \mathfrak{Re} 
        \{ \mathbf{w}^H 
        \underbrace{\mathbf{H} \mathbf{w}_t }_{\tilde{\mathbf{w}}_t}
        \}
        + {\rm const. }
\end{equation}
with equality at point $\mathbf{w}_t$.
\end{lemma}

\begin{proof}
    A concave function lies below its tangent curves, so it can be majorized at any point by its first order Taylor expansion.
\end{proof}

\begin{lemma} \label{lemma:cvxqf}
    The convex quadratic form $f:\mathbf{w} \mapsto \mathbf{w}^H \mathbf{H} \mathbf{w}$ is majorized on $\mathbb{T}_p$ at point $\mathbf{w}_t$ by the surrogate function 
    \begin{equation}
g ( \mathbf{w} | \mathbf{w}_t )  = 
2 \mathfrak{Re} 
\{ \mathbf{w}^H 
\underbrace{(\mathbf{H} - \lambda_{\rm max}^{\mathbf{H}} \mathbf{I} ) \mathbf{w}_t }_{-\tilde{\mathbf{w}}_t}
\}
+ {\rm const. }
\label{eq:linmajofquadcvx}
\end{equation}
where $\lambda_{\rm max}$ be the largest eigenvalue of $\mathbf{H}$, and with equality achieved at $\mathbf{w}_t$.
\end{lemma}

\begin{proof}
We first notice that, when $\mathbf{w}\in\mathbb{T}_p$, we have the relation
\begin{equation} \label{eq:cvx_ccv_relation}
\mathbf{w}^H (\mathbf{H} - \lambda_{\rm max}^{\mathbf{H}} \mathbf{I} ) \mathbf{w} 
=
\mathbf{w}^H \mathbf{H} \mathbf{w} -
\underbrace{p \lambda_{\rm max}^{\mathbf{H}}}_{ {\rm const.}}.
\end{equation}
Hence, the objective function restricted to $\mathbb{T}_p$ coincides with the expression of a concave quadratic form up to a constant.
Optimizing either side of the equality in \eqref{eq:cvx_ccv_relation} over $\mathbb{T}_p$ thus yields the same solution.
The quadratic form $\mathbf{w}^H (\mathbf{H} - \lambda_{\rm max}^{\mathbf{H}} \mathbf{I} ) \mathbf{w} $ is concave, thus it can be majorized at point $\mathbf{w}_t$ by its first order Taylor expansion as in \eqref{eq:linmajofquadcvx}.
\end{proof}
The two previous lemmas show that any quadratic form can be majorized on $\mathbb{T}_p$ by linear functions.
For these linear surrogates functions, the minimization step on $\mathbb{T}_p$ can be solved in closed form thanks to the following lemma.

\begin{lemma} \label{lemma:min_lin_torus}
The solution to the minimization problem
\begin{equation}
\begin{array}{c l}\label{eq:linpr_torus}
\underset{{\mathbf{w}}\in\mathbb{T}_p}{\rm minimize}
& 
- \mathfrak{Re} \left\{ \mathbf{w}^H \bar{\mathbf{w}}_t \right\}
\end{array}
\end{equation}
is obtained as
\begin{equation}
    \mathbf{w}^\star = \phi_\mathbb{T}( \bar{\mathbf{w}}_t )
\end{equation}
with the operator $\phi_\mathbb{T} : x {=} r e^{i\theta} \mapsto e^{i\theta}$ (that
extends naturally to matrices by being applied entry-wise). 
\end{lemma}
\begin{proof}
    The problem requires maximizing the sum of $p$ independent inner products in $\mathbb{C}$ of the form $\langle {e}^{i\theta_j}, [\bar{\mathbf{w}}_t]_j \rangle$, which is solved in closed-from by aligning the phase of each entry as $\theta_j = {\rm arg} ([\tilde{\mathbf{w}}_t]_j) $.
\end{proof}

\subsection{Application to KL-IPL and LS-IPL}

We can remark that LS-IPL and KL-IPL are formulated as minimization problems over $\mathbb{T}_p$ whose objective functions are quadratic forms, i.e., as instances of a generic problem
\begin{equation} \label{eq:generif_torus_qf}
    \begin{array}{c l} 
\underset{{\mathbf{w}}\in\mathbb{T}_p}{\rm minimize}
& 
\mathbf{w}^H \mathbf{M} \mathbf{w}
\end{array}
\end{equation}
where
\begin{itemize}
    \item KL-IPL involves the matrix $\mathbf{M}_{\rm KL}= \tilde{\mathbf{\Psi}}^{-1} \circ \tilde{\mathbf{\Sigma}} $, that is generally positive semi-definite (cf. remark below). 
    Thus we can use Lemma \ref{lemma:cvxqf} and Lemma \ref{lemma:min_lin_torus} in order to obtain the algorithm summarized in the box Algorithm \ref{algo:MM_KLPL}.

    \item LS-IPL involves the matrix $\mathbf{M}_{\rm LS}= - \tilde{\mathbf{\Psi}} \circ \tilde{\mathbf{\Sigma}}$, that is negative semi-definite (cf. remark below).
    Thus we can use Lemma \ref{lemma:ccvqf} and Lemma \ref{lemma:min_lin_torus} in order to obtain the algorithm summarized in the box Algorithm \ref{algo:MM_LSPL}.
    
\end{itemize}

\noindent
\textit{Remark}:
The Hadamard product of two positive semi-definite matrices is semi-definite.
However, the modulus of a positive semi-definite matrix is not guaranteed to stay positive semi-definite (some counterexamples can be found numerically).
It means that $\mathbf{M}_{\rm KL}$ (resp. $\mathbf{M}_{\rm LS}$) is not necessarily positive semi-definite (resp. negative semi-definite) by construction.
Though these corner cases were not actually observed in practice, it can be interesting to implement some safeguards checks.
As remedy, it is always possible to shift the eigenvalues to ensure positiveness, as done in Lemma \ref{lemma:cvxqf}.
The other option is to split the negative and positive part of the matrix, i.e., decomposing $\mathbf{M} = \mathbf{M}_+  + \mathbf{M}_-$, with $\mathbf{M}_+\succcurlyeq \mathbf{0}$ and $\mathbf{M}_- \preccurlyeq  \mathbf{0}$ by eigenvalue decomposition, before applying Lemma 1 and 2 separately to each resulting quadratic form.

\alglanguage{pseudocode}
\begin{algorithm}[!t]
\caption{Majorization-minimization for KL-IPL}
\label{algo:MM_KLPL}
\begin{algorithmic}[1]
\State \textbf{Entry:} $\tilde{\mathbf{\Sigma}} \in \mathbb{C}^{p\times p}$ (plug-in), $\mathbf{w}_1 \in \mathbb{T}_p$ (starting point)
\State Compute $\mathbf{M} = {\rm mod} (\tilde{\mathbf{\Sigma}} )^{-1} \circ \tilde{\mathbf{\Sigma}}  $  and $\lambda_{\rm max}^{\mathbf{M}}$
\Repeat 
\State Compute ${\bar{\mathbf{w}}_t} = (\lambda_{\rm max}^{\mathbf{M}} \mathbf{I} - \mathbf{M} ) \mathbf{w}_t $
\State Update $\mathbf{w}_t = \phi_{\mathbb{T}} \{ \bar{\mathbf{w}}_t \} $
\State $t=t+1$
\Until Convergence
\State \textbf{Output:} $\hat{\mathbf{w}}_{\boldsymbol{\theta}} = \mathbf{w}_{\rm end} \in \mathbb{T}_p$ 

\Statex
\end{algorithmic}
\vspace{-0.4cm}%
\end{algorithm}
\normalsize 

\alglanguage{pseudocode}
\begin{algorithm}[!t]
\caption{Majorization-minimization for LS-IPL}
\label{algo:MM_LSPL}
\begin{algorithmic}[1]
\State \textbf{Entry:} $\tilde{\mathbf{\Sigma}} \in \mathbb{C}^{p\times p}$ (plug-in), $\mathbf{w}_1 \in \mathbb{T}_p$ (starting point)
\State Compute $\mathbf{M} = {\rm mod}(\tilde{\mathbf{\Sigma}} ) \circ \tilde{\mathbf{\Sigma}} $
\Repeat 
\State Compute ${\bar{\mathbf{w}}_t} =  \mathbf{M} \mathbf{w}_t $
\State Update $\mathbf{w}_t = \phi_{\mathbb{T}} \{ \bar{\mathbf{w}}_t \} $
\State $t=t+1$
\Until Convergence
\State \textbf{Output:} $\hat{\mathbf{w}}_{\boldsymbol{\theta}} = \mathbf{w}_{\rm end} \in \mathbb{T}_p$ 

\Statex
\end{algorithmic}
\vspace{-0.4cm}%
\end{algorithm}

\subsection{EMI-type EVD relaxations}

The structure of the problem in \eqref{eq:generif_torus_qf} is reminiscent of the computation of an eigenvector.
Motivated by this analogy, the EMI algorithm \cite{ansari2018efficient} relaxed the constraint $\mathbf{w} \in \mathbb{T}_p$ into the following problem:
\begin{equation}   \label{eq:EMIrelax}
    \begin{array}{c l} 
\underset{{\mathbf{w}}}{\rm minimize}
& 
\mathbf{w}^H \mathbf{M} \mathbf{w} \\
{\rm subject~to} & \mathbf{w}^H \mathbf{w} = 1 .
\end{array}
\end{equation}
The solution to this problem is the eigenvector associated with the lowest (resp. largest) eigenvalue of $\mathbf{M}$ if it is hermitian positive (resp., negative) definite.
Thus it can be easily computed, and the corresponding phase estimates are then identified directly from the complex phases of the entries of this eigenvector.
In general, this relaxed estimator is not as accurate as the one brought by the proper resolution of \eqref{eq:generif_torus_qf}.
However, it has been evidenced to be satisfactory in some practical cases \cite{ansari2018efficient, bai2023lamie}, justifying its use in terms of performance versus computational load trade-off.


\section{Riemannian optimization on $\mathbb{T}_p$}
\label{sec:riem_opt}

Though majorization-minimization provides simple and elegant algorithms, its application is limited to IPL objective functions that are based on quadratic forms.
Because $\mathbb{T}_p$ \eqref{eq:phaseonlytorus} is a smooth manifold, more general objective functions can be handled with the theory of Riemannian optimization \cite{absil2008optimization, boumal2023introduction}.
This framework has been leveraged for optimization on $\mathbb{T}_p$ in beamforming application \cite{smith1999optimum, elmossallamy2021ris, xiong2023mimo}. However, to the best of our knowledge, it has not been considered within the context of IPL.
Its main interest compared to direct optimization of the argument vector $\boldsymbol{\theta}$ is that it greatly simplifies the calculations, while inherently accounting for the geometry and invariances of the considered space.
This section presents basic tools that allow developing first-order based methods such as the Riemannian counterparts of steepest descent, conjugate gradient, or Broyden–Fletcher–Goldfarb–Shanno (BFGS) algorithms.

The torus of phase-only complex vector $\mathbb{T}_p$ in \eqref{eq:phaseonlytorus}
is a compact smooth manifold embedded in $\mathbb{C}^p$.
Its tangent space at point $\mathbf{w}$ is denoted $T_{\mathbf{w}} \mathbb{T}_p$, and defined 
\begin{equation}
    \label{eq:tangent_torus}
    T_{\mathbf{w}} \mathbb{T}_p = 
    \left\{ 
    \boldsymbol{\xi} \in \mathbb{C}^p
    ~|~
    \mathfrak{Re}\{ \boldsymbol{\xi} \circ \mathbf{w}^* \} = \mathbf{0}
    \right\}
\end{equation}
which can intuitively be recovered by identifying it to a product of tangent space of $p$ complex circles (representing each entry of the vector $\mathbf{w})$.
Then $\mathbb{T}_p$ can be turned into a Riemannian submanifold \cite[Chap. 3]{boumal2023introduction} by endowing each tangent space with the Euclidean metric
\begin{equation}
    \label{eq:metric_torus}
    \begin{array}{l c c l}
         \langle \cdot , \cdot \rangle_\mathbf{w} : & T_{\mathbf{w}} \mathbb{T}_p \times T_{\mathbf{w}} \mathbb{T}_p &\rightarrow& \mathbb{R}  \\
         & \boldsymbol{\xi}, \boldsymbol{\eta} & \mapsto & \mathfrak{Re}\{ \boldsymbol{\xi}^H \boldsymbol{\eta}  \}
    \end{array} 
\end{equation}
We then consider optimization problem of the form \eqref{eq:COFI-PL2}, that has no obvious closed-form solution on $\mathbb{T}_p$.
In order to evaluate this solution, we resort to iterative
methods, i.e., methods that yield a sequence of iterates $\mathbf{w}_t\in \mathbb{T}_p$ from a starting point $\mathbf{w}_0 \in \mathbb{T}_p$.
This sequence is constructed so that it eventually converges to
a critical point of the objective in (24). 
When the variable is constrained to lie in the manifold $\mathbb{T}_p$, first-order based Riemannian optimization methods operate as follows:
\begin{enumerate}
    \item At iterate $\mathbf{w}_t$, a descent direction in the tangent space, denoted $\boldsymbol{\xi}_t\in T_{\mathbf{w}} \mathbb{T}_p$, is computed by leveraging the Riemannian gradient of the objective function.

    \item The direction descent  $\boldsymbol{\xi}_t$ is used to obtain the next iterate $\mathbf{w}_{t+1}$. This is achieved through a retraction on $\mathbb{T}_p$, which is an operator that maps tangent vectors back onto the manifold. 
\end{enumerate}
This generic procedure is illustrated in Figure \ref{fig:riem_opt_torus}.

For the first step, the steepest direction descent is given by the Riemannian gradient, which is defined according to the metric \eqref{eq:metric_torus}.
The Riemannian gradient of a function $f: \mathbb{T}_p \rightarrow \mathbb{R} $ at point $\mathbf{w}_t$ is the unique vector defined as
\begin{equation}
   \langle {\rm grad} f(\mathbf{w}_t) , \boldsymbol{\xi} \rangle_{\mathbf{w}_t} = 
   {\rm D} f(\mathbf{w}_t ) \left[ \boldsymbol{\xi} \right],
\end{equation}
where ${\rm D} f(\mathbf{w}_t ) \left[ \boldsymbol{\xi} \right]$ is the directional derivative of $f$ with respect to $\mathbf{w}_t$ in direction $\boldsymbol{\xi} $.
This directly yields the Riemannian gradient as the orthogonal projection of the Euclidean gradient onto the tangent space, i.e.:
\begin{equation} \label{eq:riemgrad}
     {\rm grad} f(\mathbf{w}_t)
     = 
    \nabla f(\mathbf{w}_t) - \mathfrak{Re} \{ \nabla f(\mathbf{w}_t)^* \circ \mathbf{w}_t \} \circ \mathbf{w}_t,
\end{equation}
where $\nabla f(\mathbf{w}_t)$ is the Euclidean gradient of the objective
function at $\mathbf{w}_t$.
Momentum-based, or acceleration methods generally combine the gradient of several iterates.
Because these objects belong to different tangent spaces in our context, we need 
to meaningfully map tangent vectors from one point to another.
This is achieved thanks to parallel transport \cite[ref section]{boumal2023introduction}, which is obtained in our context with 
\begin{equation}
    \mathcal{T}_{\mathbf{w}_{t}\leftarrow\mathbf{w}_{t-1}}^{\mathbb{T}_p}
    (\boldsymbol{\xi})
    = \boldsymbol{\xi} - \mathfrak{Re}( \boldsymbol{\xi} \circ \mathbf{w}_t^* ) \circ \mathbf{w}_t.
\end{equation}
where $\mathcal{T}_{\mathbf{w}_{t}\leftarrow\mathbf{w}_{t-1}}^{\mathbb{T}_p}$ is used to denote the transport of a tangent vector $\boldsymbol{\xi}\in T_{\mathbf{w}_{t-1}} \mathbb{T}_p$ to the tangent space $T_{\mathbf{w}_{t}} \mathbb{T}_p$.

Given a direction descent $\boldsymbol{\xi}_t \in T\mathbf{w}_t\mathbb{T}_p $, 
the second step can be performed by defining a retraction, i.e., an operator $\mathcal{R}_{\mathbf{w}_t}: T\mathbf{w}_t \mathbb{T}_p \to \mathbb{T}_p $ that maps tangent vectors back to the manifold.
In our case, the euclidean projection on $\mathbb{T}_p$ is a practical candidate, so we can use
\begin{equation} \label{eq:retraction}
    \mathcal{R}_{\mathbf{w}_t} ({\boldsymbol{\xi}}_t) = \phi_\mathbb{T}(\mathbf{w}_t + \boldsymbol{\xi}_t ),
\end{equation}
where $\phi_\mathbb{T}$ was defined after \eqref{eq:mod_arg_decomp}.

The gradient, parallel transport, and retraction are enough to derive most first-order based algorithms by transposing from standard optimization methods.
As an example, the Riemannian conjugate gradient is given in the box Algorithm \ref{algo:Riemconjgrad}, while the Euclidean gradient for covariance fitting cost computed with LS and KL can be obtained from the relations:
\begin{equation}
    \nabla  \mathbf{w}^H \mathbf{M} \mathbf{w}  = 2 \mathbf{M} \mathbf{w},
\end{equation}
for all $\mathbf{M}\in \mathcal{H}_p$. For the WLS cost, the derivation is a little bit tricky and leads to 
\begin{equation}
\nabla f^{\rm WLS}_{\tilde{\mathbf{\Sigma}}} = 4 {\rm diag}(\mathbf{M} {\rm diag}(\mathbf{w})^H \mathbf{M}-\mathbf{M}),
\end{equation}
where $\mathbf{M}=\tilde{\boldsymbol{\Sigma}}^{-1}{\rm diag}(\mathbf{w})\tilde{\boldsymbol{\Sigma}}$ and {\rm diag} gives a diagonal matrix when the entry is a vector and a vector when the entry is a diagonal matrix. 

\alglanguage{pseudocode}
\begin{algorithm}[!t]
\caption{Riemannian conjugate gradient on $\mathbb{T}_p$}
\label{algo:Riemconjgrad}
\begin{algorithmic}[1]
\State \textbf{Entry:} Objective $f$ of problem \eqref{eq:COFI-PL2}, starting point $\mathbf{w}_0$
\Repeat 
\State Compute $\boldsymbol{\eta}_t = - {\rm grad} f(\mathbf{w}_t)$ with \eqref{eq:riemgrad}
\State Compute $\alpha_t$ and $\beta_t$ with \cite{hestenes1952methods}
\State Set direction $\boldsymbol{\xi}_t = \alpha_t \boldsymbol{\eta}_t + \beta_t \mathcal{T}_{\mathbf{w}_{t}\leftarrow\mathbf{w}_{t-1}}^{\mathbb{T}_p} (  \boldsymbol{\eta}_{t-1}) $
\State Update $\mathbf{w}_t = \mathcal{R}_{\mathbf{w}_t} ({\boldsymbol{\xi}}_t)$ with \eqref{eq:retraction}
\Until Convergence
\State \textbf{Output:} $\hat{\mathbf{w}}_{\boldsymbol{\theta}} = \mathbf{w}_{\rm end} \in \mathbb{T}_p$ 

\Statex
\end{algorithmic}
\vspace{-0.4cm}%
\end{algorithm}


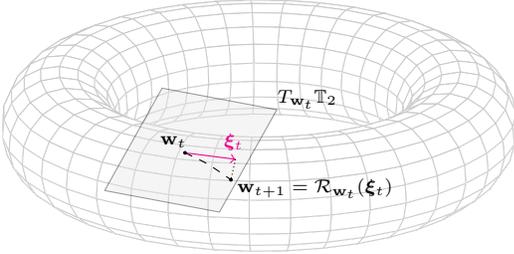
\begin{figure}[!t]
    \centering
    \begin{tikzpicture}
        
    \begin{axis}[
        axis equal image,
        hide axis,
        z buffer = sort,
        view = {0}{20},
        scale = 1
        ]
        
        \addplot3[
            surf,
            shader    = faceted interp,
            samples   = 20,
            samples y = 40,
            domain    = 0:2*pi,
            domain y  = 0:2*pi,
            colormap name = justblackandwhite, thin
        ] (
            {(3.5+sin(deg(\x)))*cos(deg(\y))},
            {(3.5+sin(deg(\x)))*sin(deg(\y))},
            {cos(deg(\x))}
        );

    	\addplot3[mark=*,mark size=0.5pt, only marks,draw=black,fill=black] coordinates {
        (-1.3, 0, -0.5) 
        (-0.5,    0, -1) 
	    };

        \draw [fill=gray!20,opacity=0.4] 
        (-2.7,0,-1.2) -- (-0.7,0,-1.6) -- (0.3,0,0.3) -- (-1.7,0,0.7) -- (-2.7,0,-1.2);

        \draw [dashed]
        (-1.3, 0, -0.5) .. controls (-0.9,0,-0.7) ..  (-0.5,    0, -1);

       \draw [->, color = magenta] (-1.3, 0, -0.5) -- (-0.42, 0, -0.62);

       \draw [densely dotted] (-0.42, 0, -0.62) -- (-0.5,    0, -1);

    \node[font=\scriptsize, anchor=south west] at (0.15,0,0.15) {${T}_{\mathbf{w}_t} \mathbb{T}_2$};
    \node[font=\scriptsize] at (-1.5, 0, -0.3) {$\mathbf{w}_t$};
    \node[font=\scriptsize] at (-0.45, 0, -0.3) {\textcolor{magenta}{$\boldsymbol{\xi}_t$}};
    \node[font=\scriptsize, anchor=north west] at (-0.55, 0,-0.8) {$\mathbf{w}_{t+1}=\mathcal{R}_{\mathbf{w}_t}(\boldsymbol{\xi}_t) $};
    
    \end{axis}

    \end{tikzpicture}
    \caption{Illustration of Riemannian optimization on $\mathbb{T}_2$ (represented as a torus embedded in $\mathbb{R}^3$): the iterate $\mathbf{w}_{t+1}$ is obtained from the retraction $\mathcal{R}_{\mathbf{w}_t}$ applied to the direction descent $\boldsymbol{\xi}_t \in {T}_{\mathbf{w}_t} \mathbb{T}_2$ (i.e., a vector of the tangent space of $\mathbb{T}_2$ at point $\mathbf{w}_t$).}
    \label{fig:riem_opt_torus}
\end{figure}


\section{Links with existing algorithms}

\subsection{Existing algorithms as COFI-PL instances}

\label{sec:sota}

An overview relating state-of-the-art IPL algorithms to module options of COFI-PL is presented in Table \ref{tab:overview}.
We can notice the prevalence of KL-IPL formulations due to their interesting Gaussian maximum likelihood roots.
In this scope, a subtlety requires some discussion:
as it is formalized in this paper, COFI-PL aims at fitting the phases of the plug-in (regularized) estimator $\tilde{\mathbf{\Sigma}}$ to its own modulus $\tilde{\mathbf{\Psi}} \overset{\Delta}{=} {\rm mod}(\tilde{\mathbf{\Sigma}})$.
For KL-IPL, this means minimizing the quadratic form on $\mathbb{T}_p$ in \eqref{eq:generif_torus_qf}, with the matrix  $\mathbf{M}_{\rm KL}= \tilde{\mathbf{\Psi}}^{-1} \circ \tilde{\mathbf{\Sigma}} $.
On the other hand, traditional formulations of KL-IPL arise from a Gaussian maximum likelihood approach.
In this case, $\tilde{{\mathbf{\Sigma}}} = \mathbf{S}$ is set by the construction of the likelihood function, and the regularization is only applied to the plug-in estimate of the unknown modulus $\tilde{\mathbf{\Psi}}$ (before inversion).
Such construction of the objective function is denoted as KL$_{\rm ML}$ in the table to highlight the distinction.
Though not studied in this paper for the sake of clarity, this difference suggests a more flexible implementation, in which one can dissociate (mix-and-match) regularization processes applied to $\tilde{\mathbf{\Psi}}^{-1} $ and $\tilde{\mathbf{\Sigma}}$ when constructing $\mathbf{M}_{\rm KL}$.

\begin{table*}[!t]
    \centering
    \begin{tabular}{|c|c|c|c|c|c|}
        \hline
        Method &  Plug-in $\hat{\mathbf{\Sigma}}$ & Regularization & Fitting cost & Optimization & Remark \\ \hline
        
        PL \cite{Guarnieri2008} & SCM & & KL & Elementwise on $\{ \theta_q \}_{q=1}^p$ & \\

        PTA \cite{Ferretti2011_squee} & SCM & & KL & BFGS on $\boldsymbol{\theta}$ & \\

        CAESAR \cite{Fornaro2015}  & SCorr  & LR \eqref{eq:low_rank_approx}, $r=1$ & KL$_{\rm ML}$ &   &  \\ 

        Cao et al. \cite{Cao2015} & SCorr  & LR (optional) & KL &  & LR uses KL$_{\rm ML}$ as in \cite{Fornaro2015}\\ 

        Bootstrapping \cite{jiang2020distributed} & Boot. Corr &  & KL$_{\rm ML}$ &  & Bootstrapping method for coherence estimation\\
        
        EMI  \cite{ansari2018efficient} & SCM &  & KL & EVD relax. & \\

        
        TMLE \cite{wang2022new} & SCM  & shrink. & \eqref{eq:mindetTMLE} & & \\ 

        Zwieback \cite{zwieback2022cheap} & SCM & shrink. and/or Tap.   & KL$_{\rm ML}$  &  & \\

        LS-PL \cite{Vu2023igarss} & SCM & & LS &  MM on $\mathbf{w}_{\boldsymbol{\theta}}$ & \\

        LaMIE \cite{bai2023lamie} & SCM & shrink. and Tap.  & LS & BFGS on $\boldsymbol{\theta}$ or EVD relax. & Parameter selection with \cite{ollila2022regularized}\\\hline
    \end{tabular}
    \vspace{0.2cm} 
    \caption{State-of-the-art IPL algorithms as instances of COFI-PL}
    \label{tab:overview}
\end{table*}

\subsection{Other maximum likelihood approaches}

Section \ref{sec:estimreguljoint} discussed the use of structure constrained covariance matrix estimators.
It naturally raises the following questions: is it possible to directly perform structured maximum likelihood estimation with the structure constraint in \eqref{eq:cov_struct_1}?
Unfortunately, the modulus-argument decomposition is not an holomorphic function \cite{rudin2006real,appel2007mathematics}, so it is not well suited to a differentiation-based optimization process.
To overcome this issue, several works \cite{Vu2022igarss, wang2022new, Vu2023robust} considered the alternative decomposition $\mathbf{\Sigma} = \boldsymbol{\Psi}_{\mathcal{R}}  \circ (\mathbf{w}_{\boldsymbol{\theta}} \mathbf{w}_{\boldsymbol{\theta}}^H)$, where $\boldsymbol{\Psi}_{\mathcal{R}} \in \mathcal{S}_p^{++}$ is a real positive definite matrix.
This decomposition involves parameters in smooth manifolds, so it can be efficiently accounted for in a constrained optimization problem.
It also coincides with the modulus-argument when all entries in $\boldsymbol{\Psi}_{\mathcal{R}}$ are positive.
However, it is less restrictive, and the resulting solutions can potentially include negative entries in $\boldsymbol{\Psi}_{\mathcal{R}}$, i.e., some ambiguities that break the phase closure property (possible correction procedures are discussed in \cite{Vu2023robust}).
In the Gaussian case, \cite{wang2022new} relies on a relaxing simplification that allows to compress the likelihood in the objective function into a single determinant:
\begin{equation} \label{eq:mindetTMLE}
    \begin{array}{c l} 
\underset{{\mathbf{w}_{\boldsymbol{\theta}}}\in\mathbb{T}_p}{\rm minimize}
& 
\left| \mathfrak{Re} \{ {\rm diag}(\mathbf{w}_{\boldsymbol{\theta}})^H \tilde{\mathbf{S}} {\rm diag}(\mathbf{w}_{\boldsymbol{\theta}}) \}  \right|
\end{array}
\end{equation}
Alternatively, majorization-minimization algorithms for the joint optimization of $\boldsymbol{\Psi}_{\mathcal{R}} $ and $ \mathbf{w}_{\boldsymbol{\theta}}$ are derived in \cite{Vu2022igarss} for the Gaussian case, and extended to rank-constrained and/or scaled Gaussian models in \cite{Vu2023robust}.
In any case, the term ``true maximum likelihood'' needs to be handled with caution because \cite{wang2022new} appears to employ a relaxation of the Gaussian log-likelihood. Moreover, all works \cite{Vu2022igarss, wang2022new, Vu2023robust} rely on a matrix decomposition that is more permissive than the expected structure \eqref{eq:cov_struct_1} (though evidenced to be useful in practice). 
\section{Real case study: subsidence of Mexico city}
\label{sec:simu}

\begin{figure*}[htb]
    \begin{minipage}{0.32\textwidth}
        \centering
        \includegraphics[width=1.0\columnwidth]{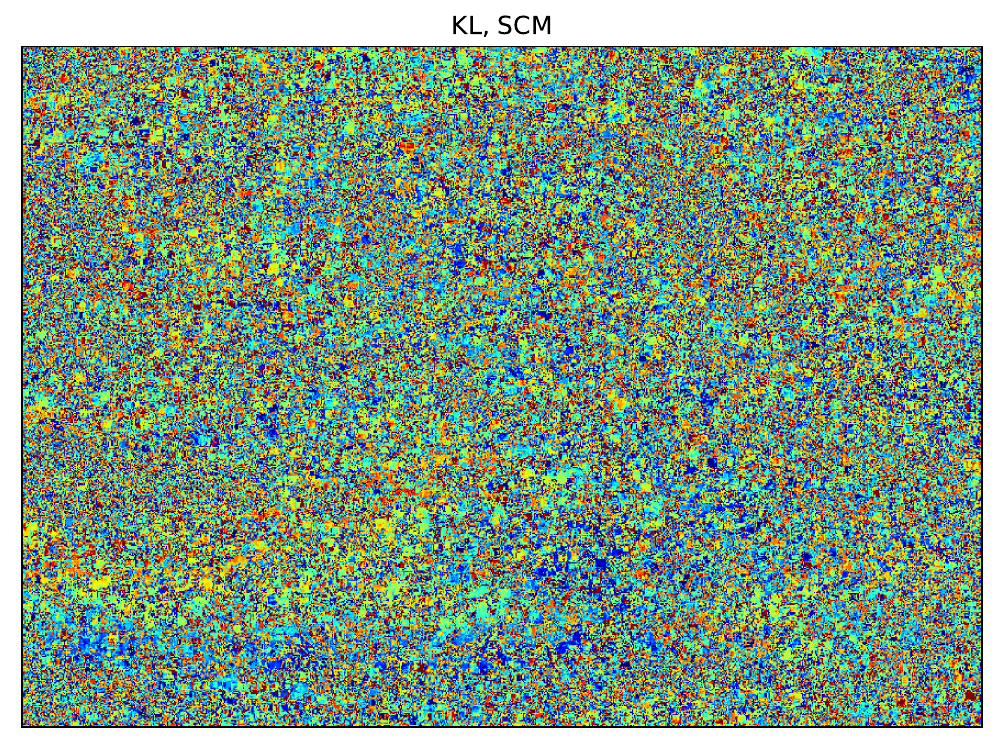}
    \end{minipage}
    \hfill{}
    \begin{minipage}{0.32\textwidth}
        \centering
        \includegraphics[width=1.0\columnwidth]{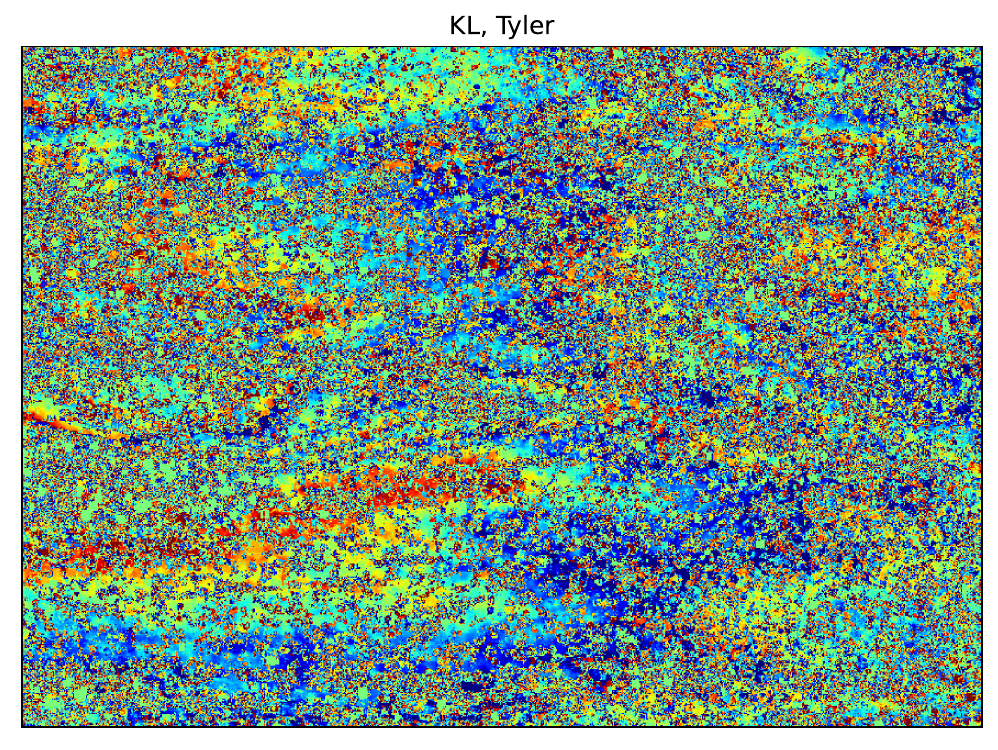}
    \end{minipage}
    \hfill{}
    \begin{minipage}{0.32\textwidth}
        \centering
        \includegraphics[width=1.0\columnwidth]{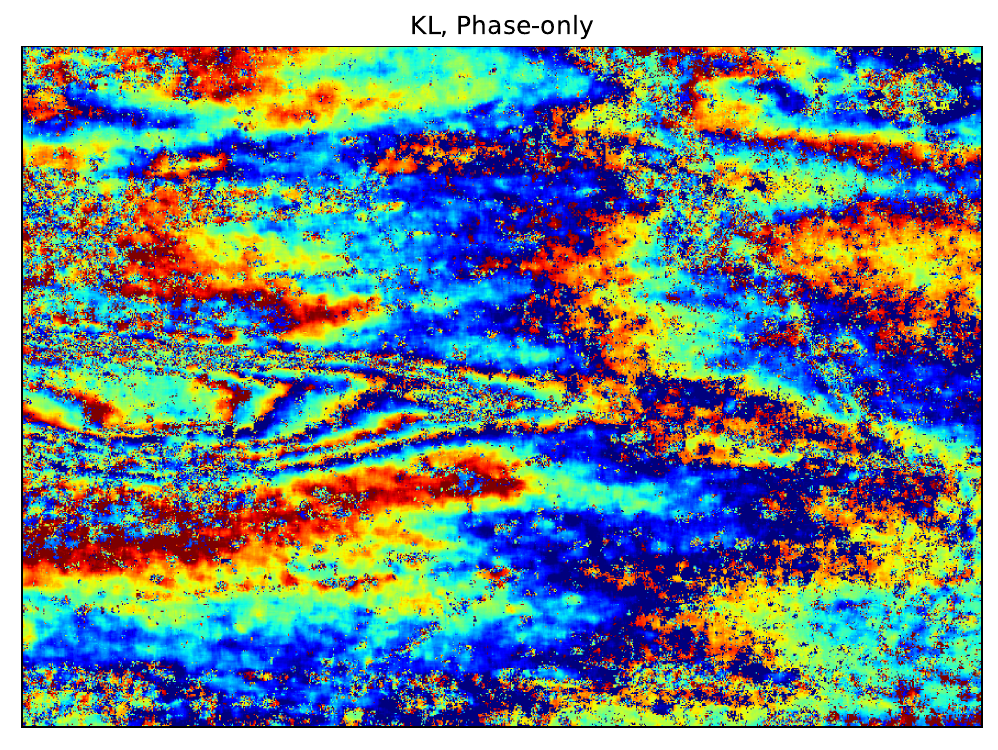}
    \end{minipage}
    \begin{minipage}{0.32\textwidth}
        \centering
        \includegraphics[width=1.0\columnwidth]{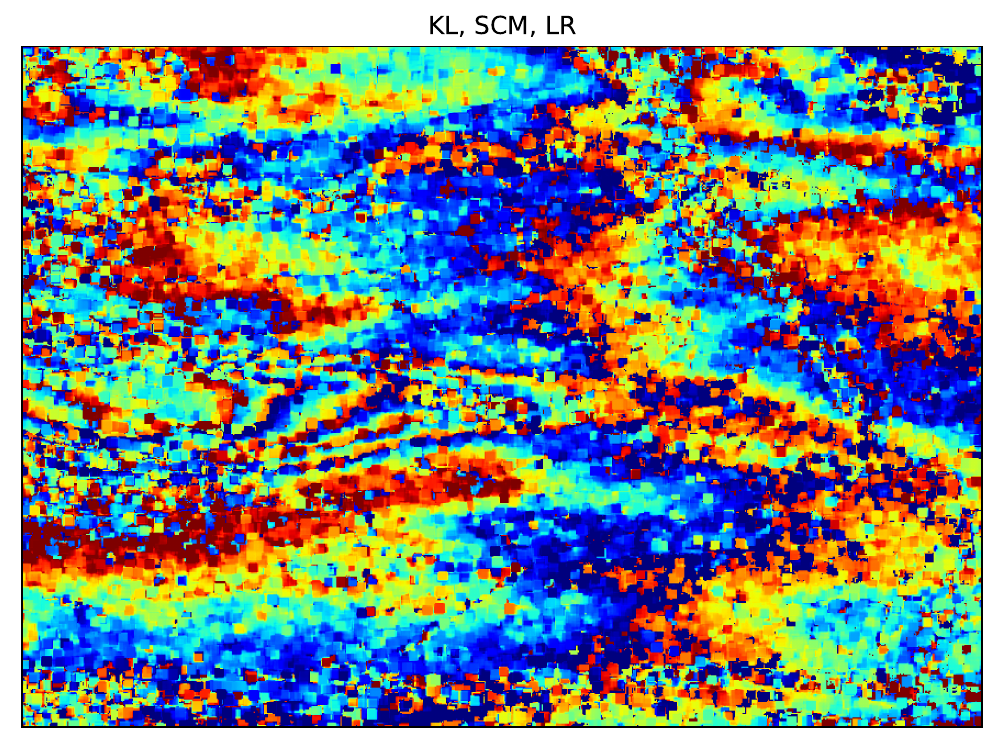}
    \end{minipage}
    \hfill{}
    \begin{minipage}{0.32\textwidth}
        \centering
        \includegraphics[width=1.0\columnwidth]{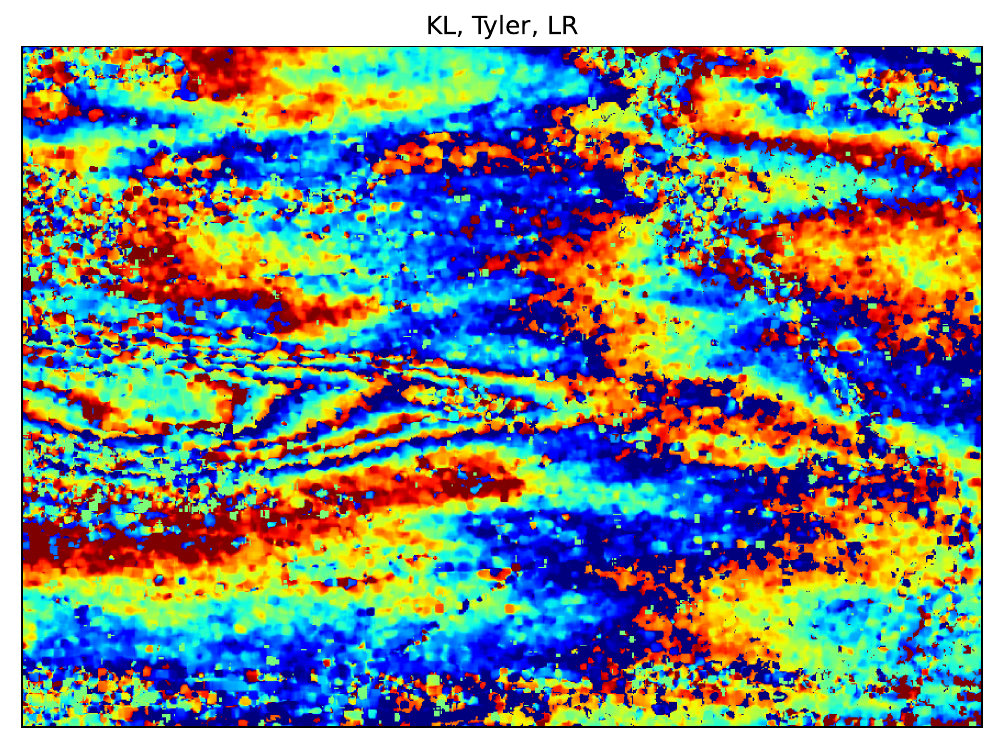}
    \end{minipage}
    \hfill{}
    \begin{minipage}{0.32\textwidth}
        \centering
        \includegraphics[width=1.0\columnwidth]{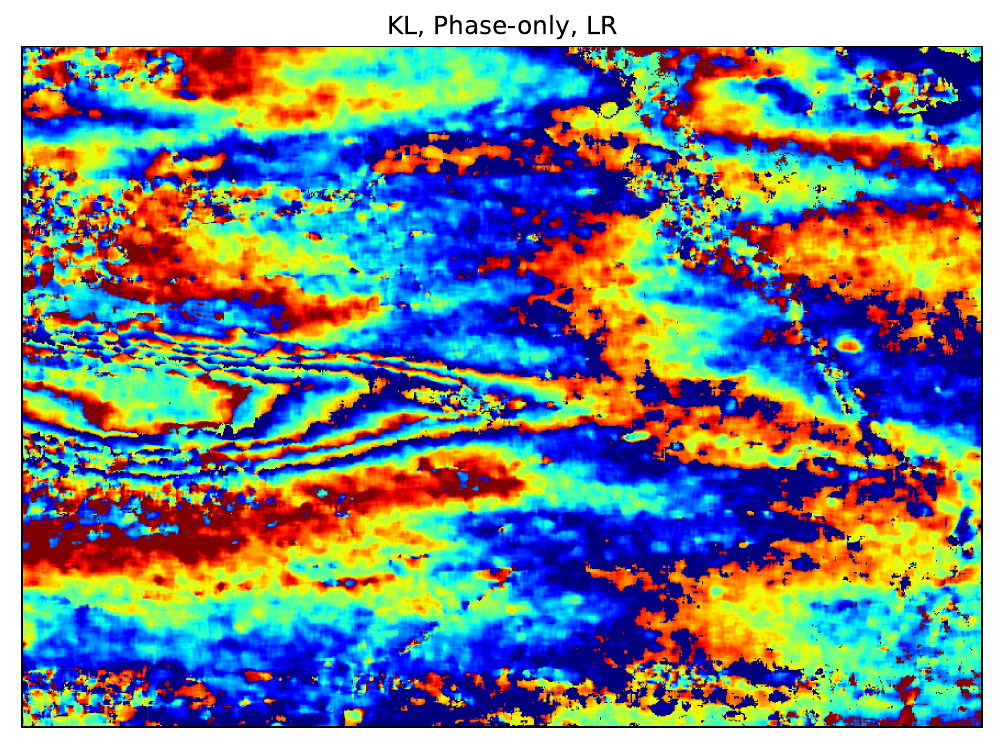}
    \end{minipage}
    \begin{minipage}{0.32\textwidth}
        \centering
        \includegraphics[width=1.0\columnwidth]{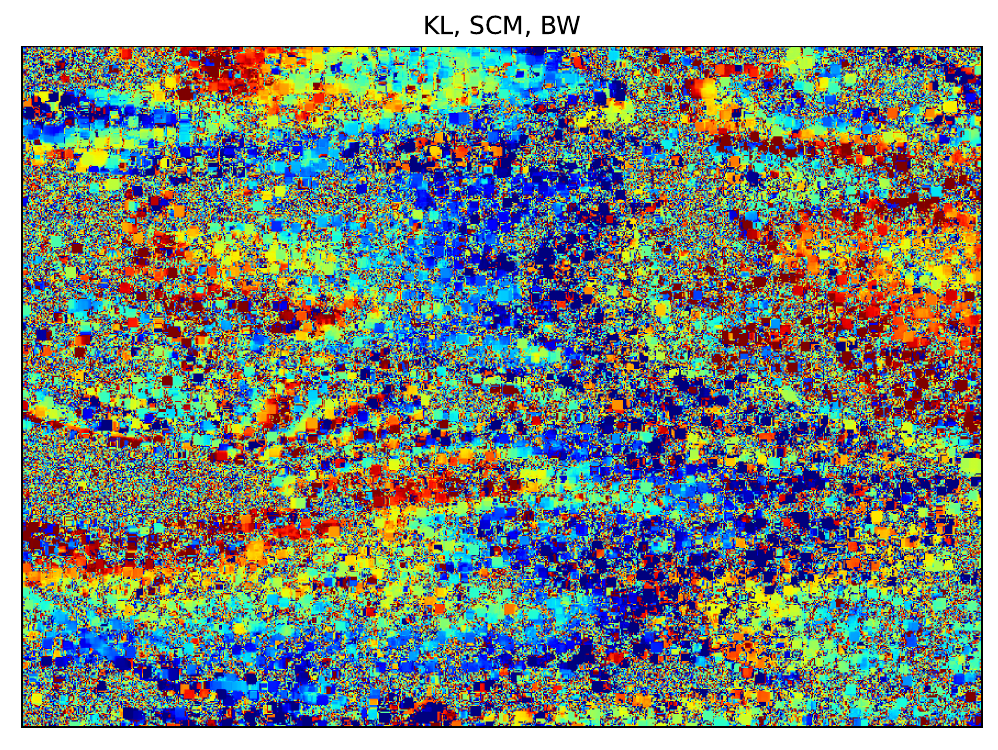}
    \end{minipage}
    \hfill{}
    \begin{minipage}{0.32\textwidth}
        \centering
        \includegraphics[width=1.0\columnwidth]{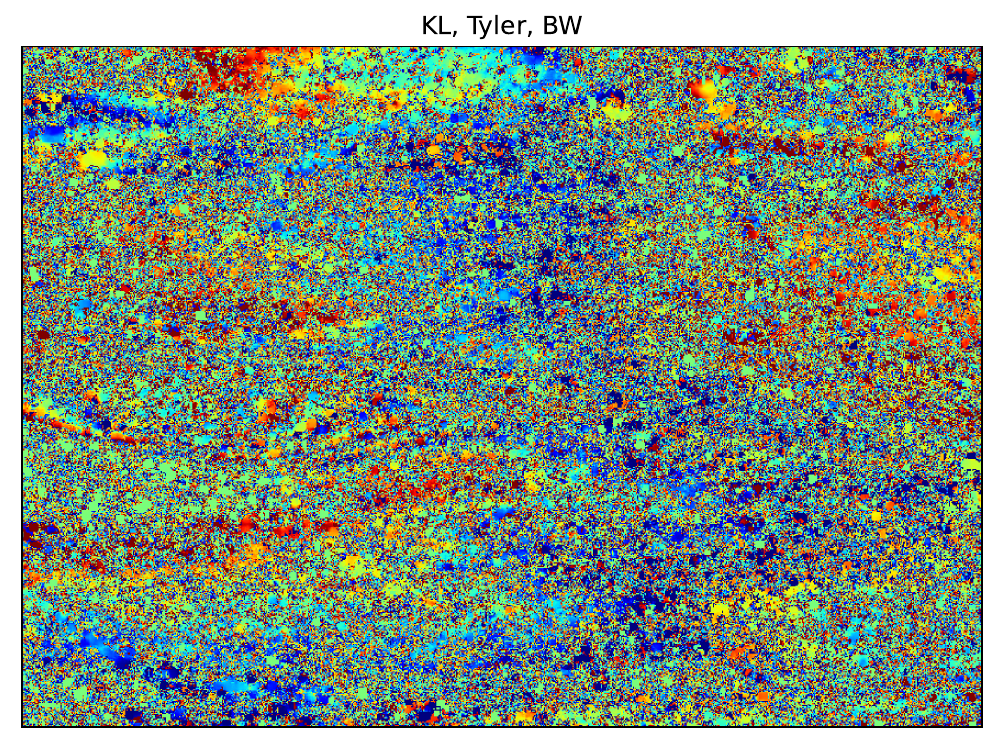}
    \end{minipage}
    \hfill{}
    \begin{minipage}{0.32\textwidth}
        \centering
        \includegraphics[width=1.0\columnwidth]{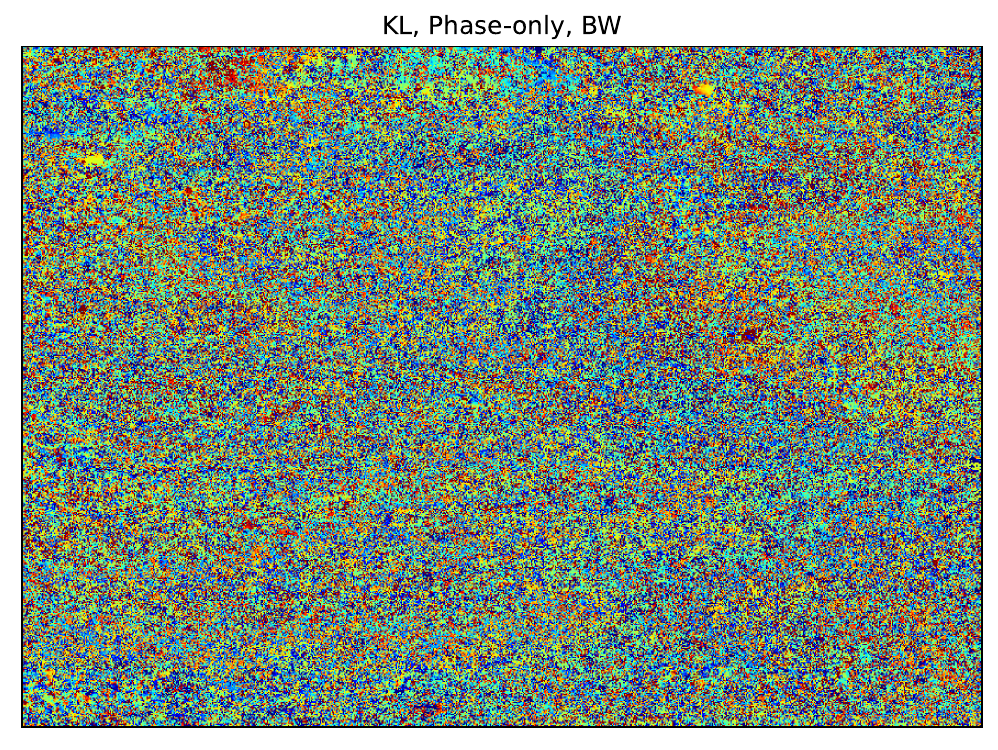}
    \end{minipage}
    \begin{minipage}{0.32\textwidth}
        \centering
        \includegraphics[width=1.0\columnwidth]{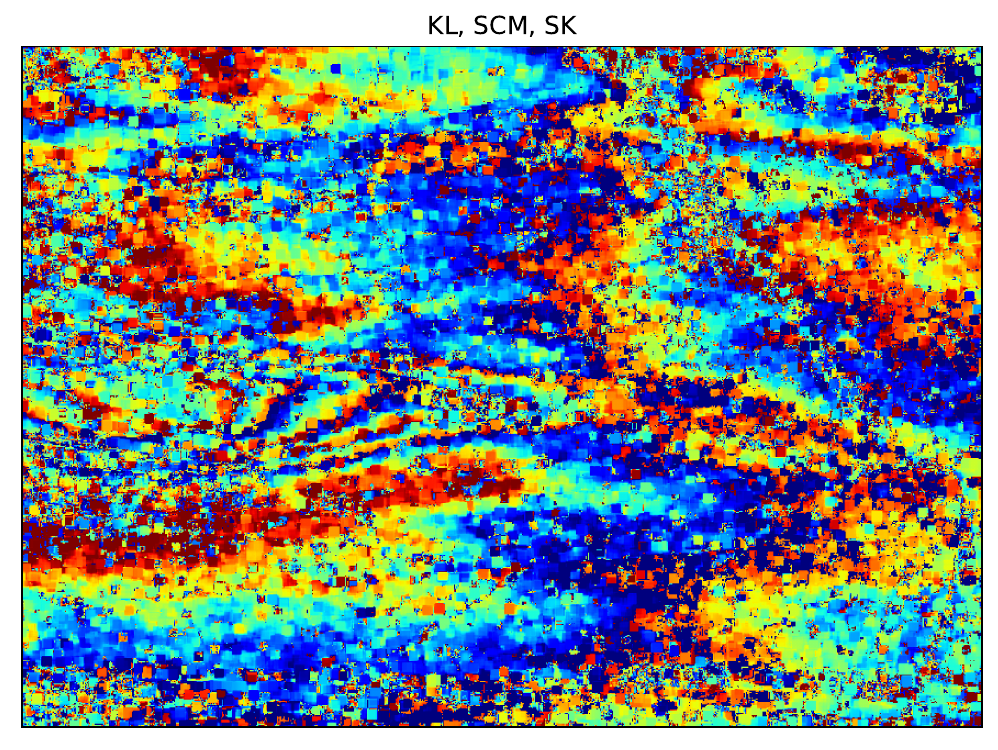}
    \end{minipage}
    \hfill{}
    \begin{minipage}{0.32\textwidth}
        \centering
        \includegraphics[width=1.0\columnwidth]{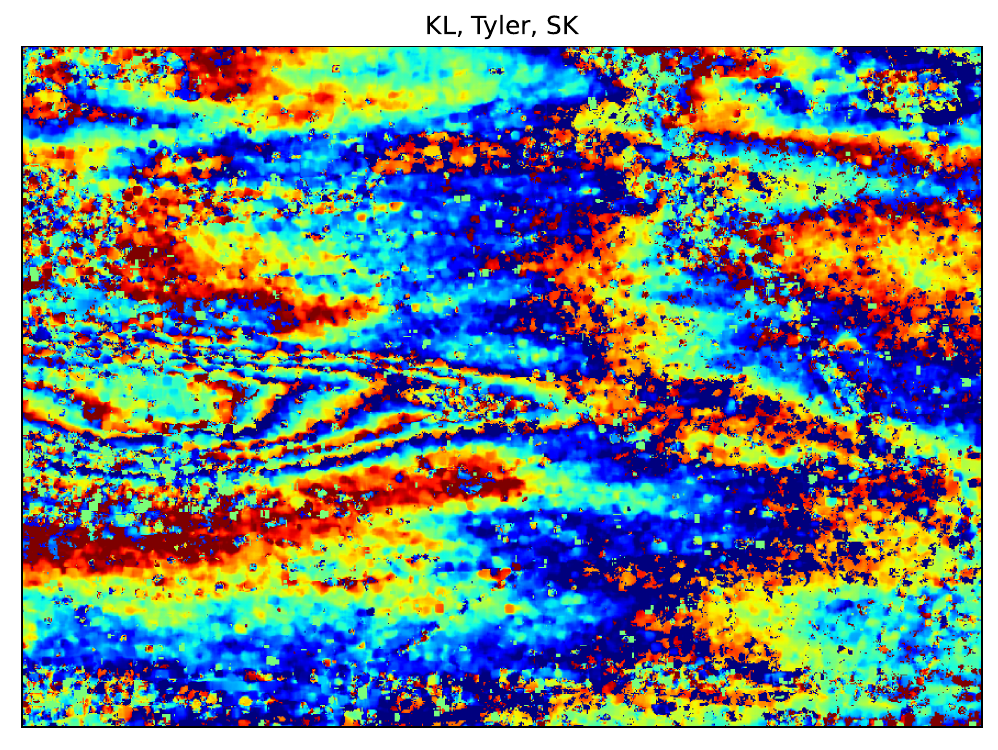}
    \end{minipage}
    \hfill{}
    \begin{minipage}{0.32\textwidth}
        \centering
        \includegraphics[width=1.0\columnwidth]{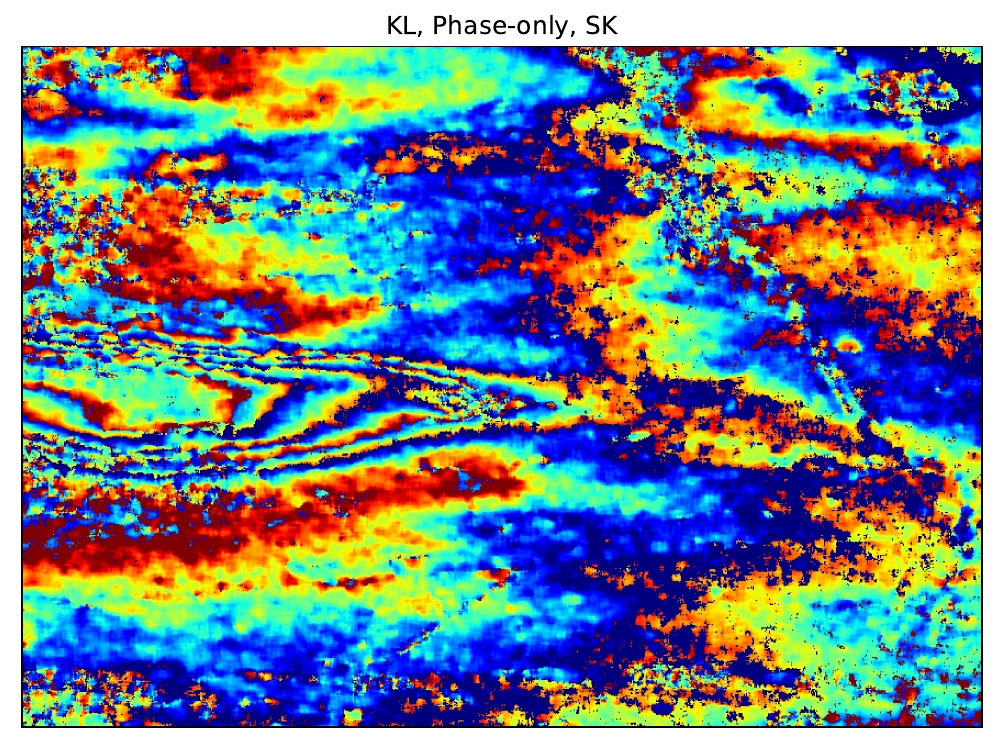}
    \end{minipage}   

    \begin{minipage}{1.0\linewidth}
    \centering
        \includegraphics[width=0.6\columnwidth]{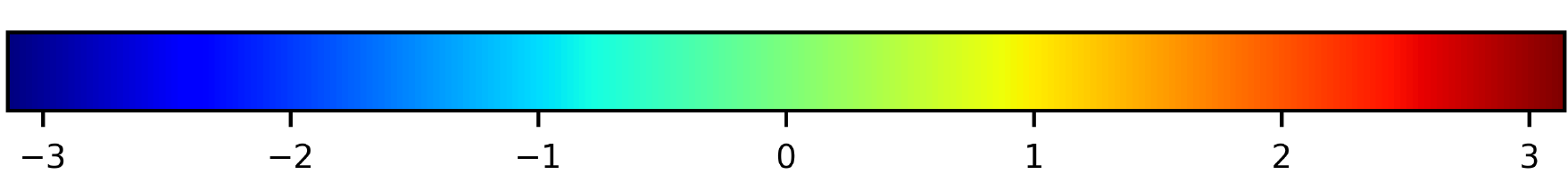}
    \end{minipage}
\caption{COFI-PL with KL fitting applied to Mexico InSAR dataset with various plug-in estimators (from left to right: sample covariance matrix \eqref{eq:SCM}, Tyler's estimator \eqref{eq:mest} with $u_{T}(t) = p/t$, phase-only sample covariance matrix \eqref{eq:phase_only_SCM}) and various regularization processes (from top to bottom: no regularization, low-rank approximation \eqref{eq:low_rank_approx2} with $k=1$, tapering \eqref{eq:tapering} with bandwidth $b=9$, shrinkage to identity \eqref{eq:LWshrink} with $\beta=0.1$).
} \label{fig:KL}
\end{figure*}

\begin{figure*}[htb]
    \begin{minipage}{0.32\textwidth}
        \centering
        \includegraphics[width=1.0\columnwidth]{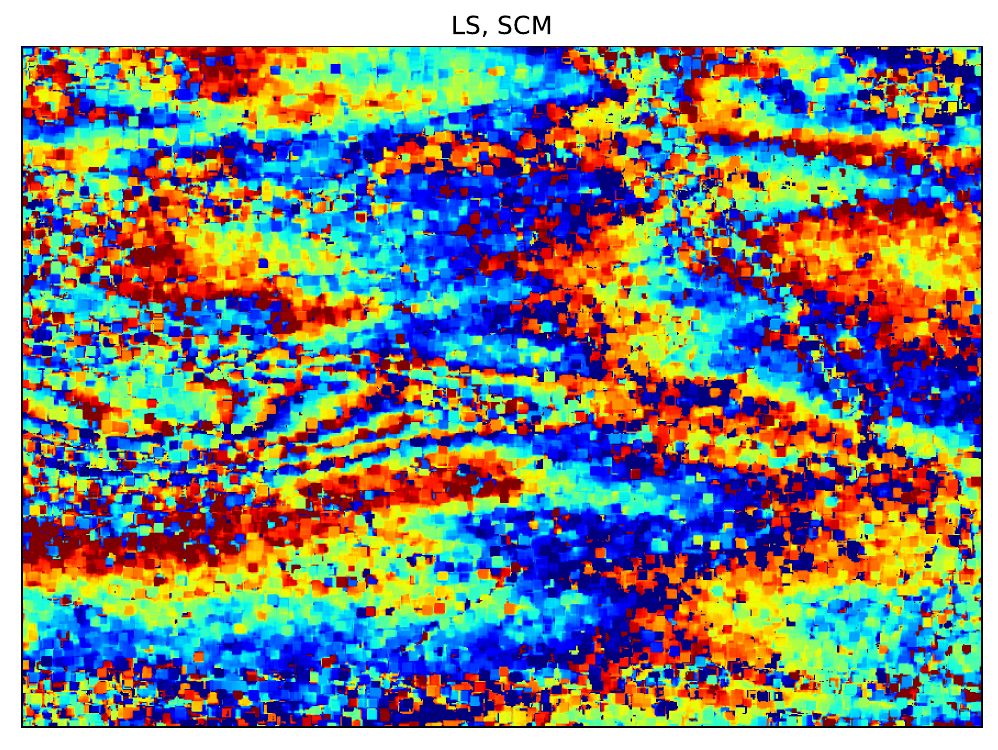}
    \end{minipage}
    \hfill{}
    \begin{minipage}{0.32\textwidth}
        \centering
        \includegraphics[width=1.0\columnwidth]{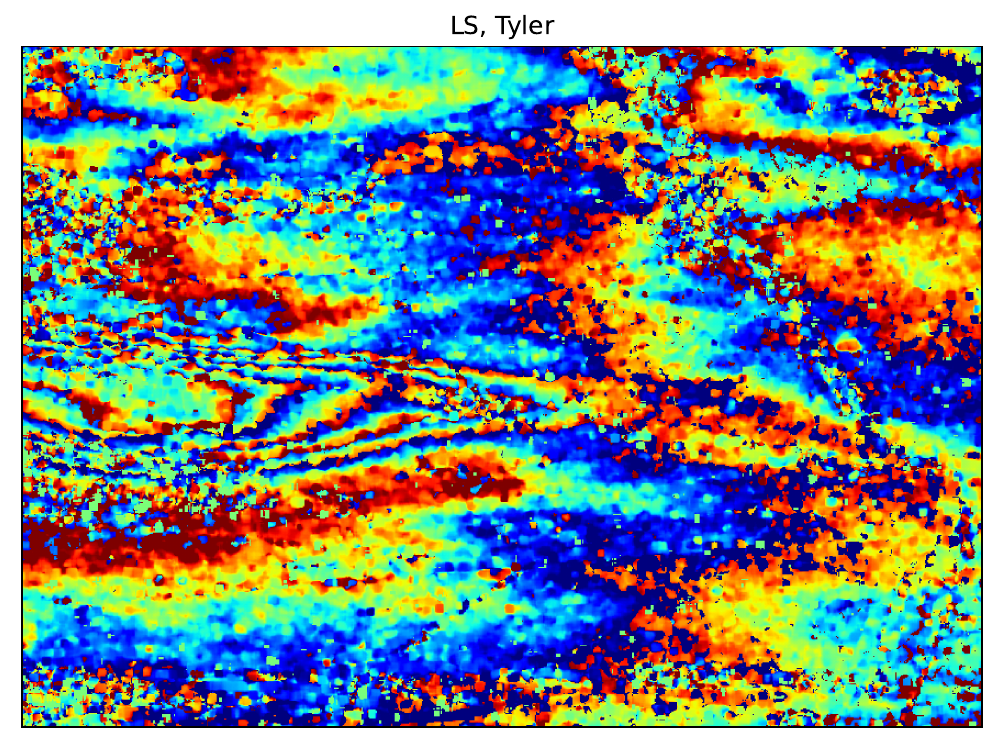}
    \end{minipage}
    \hfill{}
    \begin{minipage}{0.32\textwidth}
        \centering
        \includegraphics[width=1.0\columnwidth]{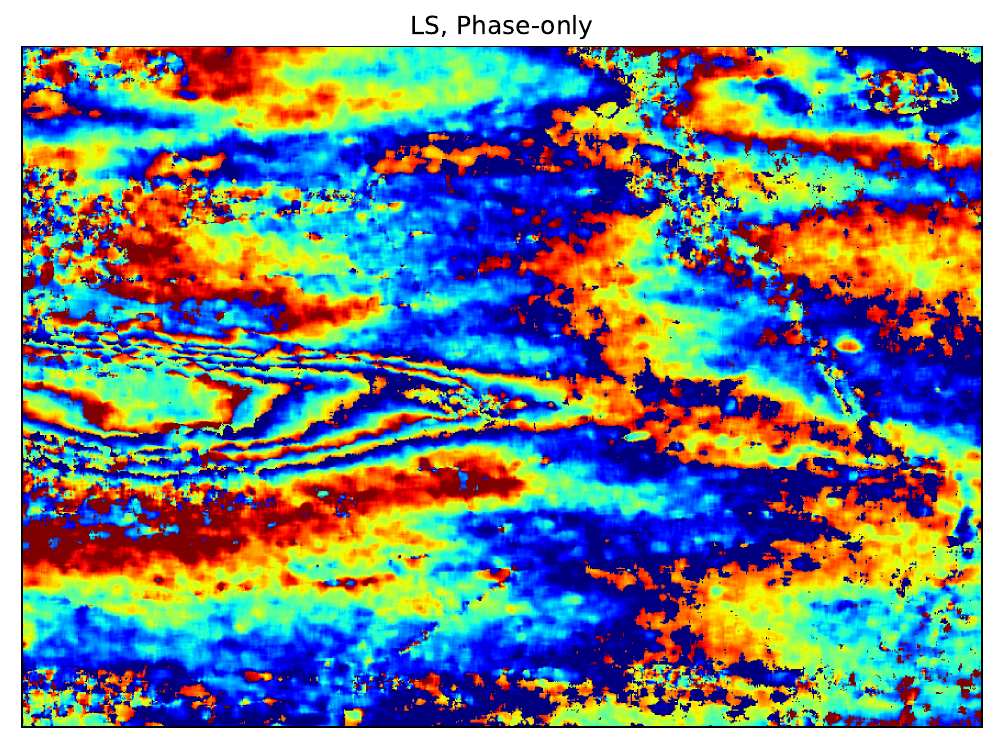}
    \end{minipage}
    \begin{minipage}{0.32\textwidth}
        \centering
        \includegraphics[width=1.0\columnwidth]{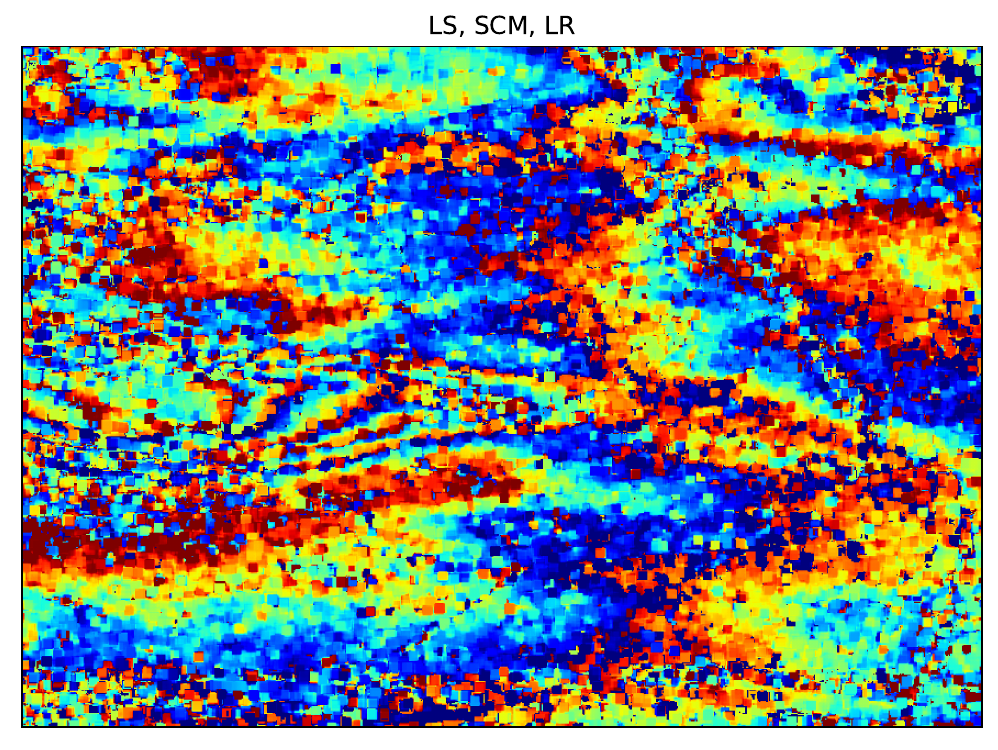}
    \end{minipage}
    \hfill{}
    \begin{minipage}{0.32\textwidth}
        \centering
        \includegraphics[width=1.0\columnwidth]{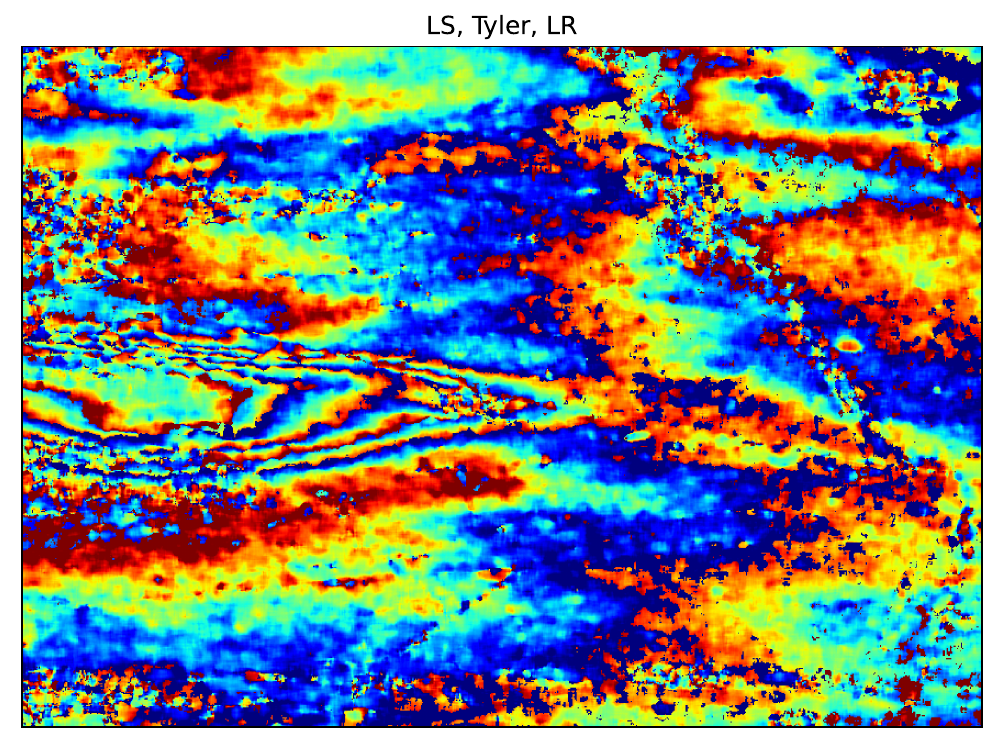}
    \end{minipage}
    \hfill{}
    \begin{minipage}{0.32\textwidth}
        \centering
        \includegraphics[width=1.0\columnwidth]{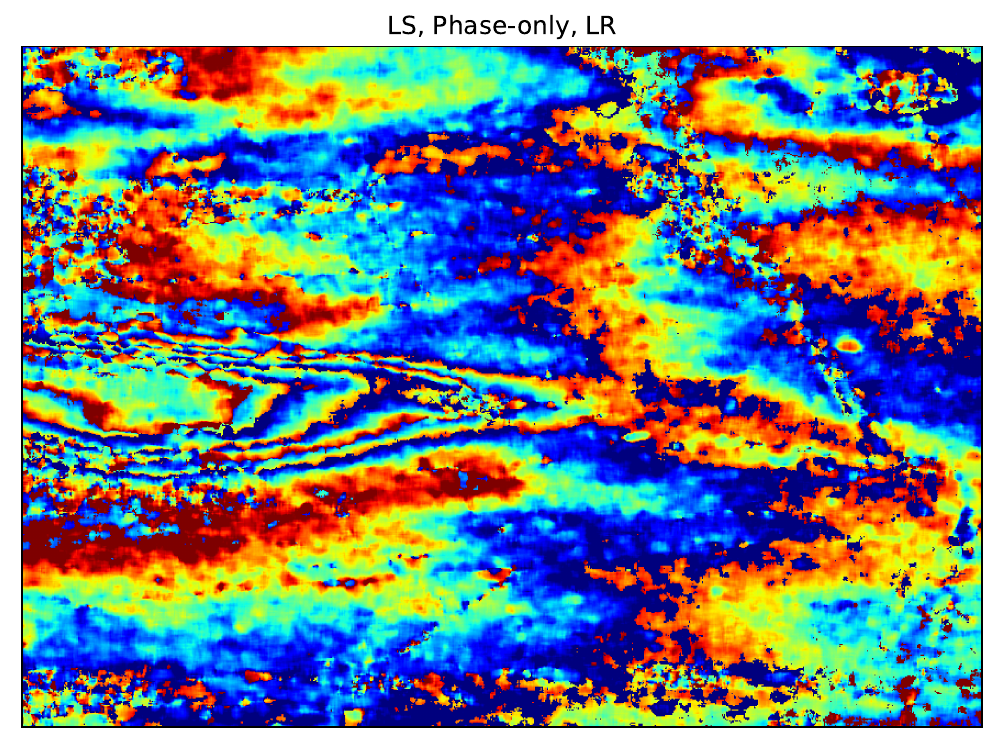}
    \end{minipage}
    \begin{minipage}{0.32\textwidth}
        \centering
        \includegraphics[width=1.0\columnwidth]{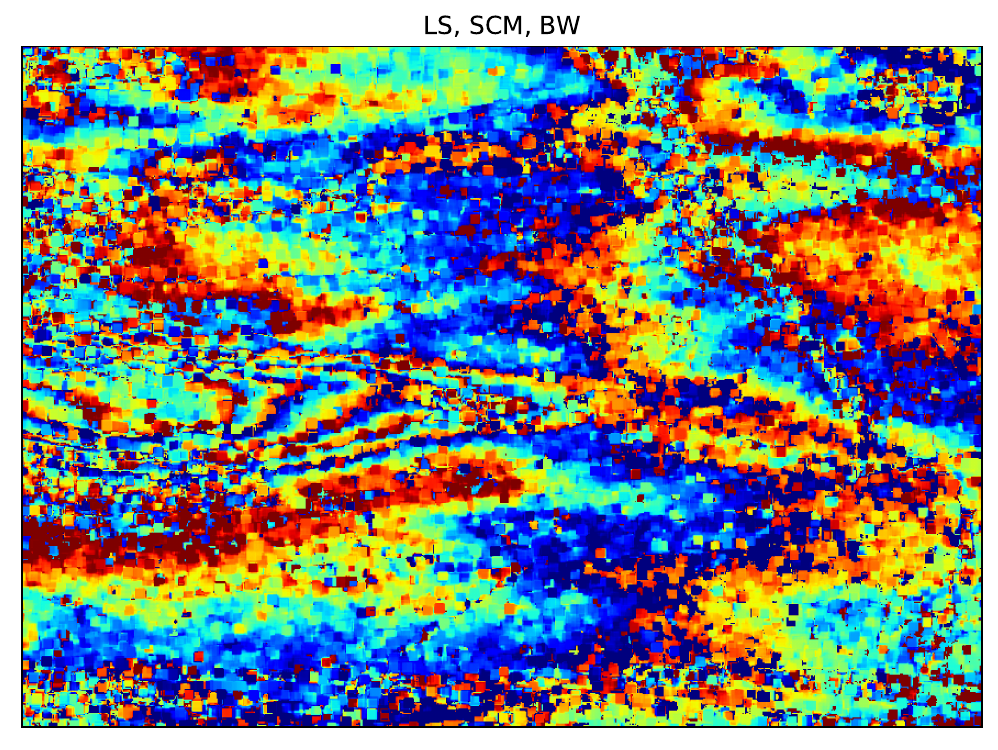}
    \end{minipage}
    \hfill{}
    \begin{minipage}{0.32\textwidth}
        \centering
        \includegraphics[width=1.0\columnwidth]{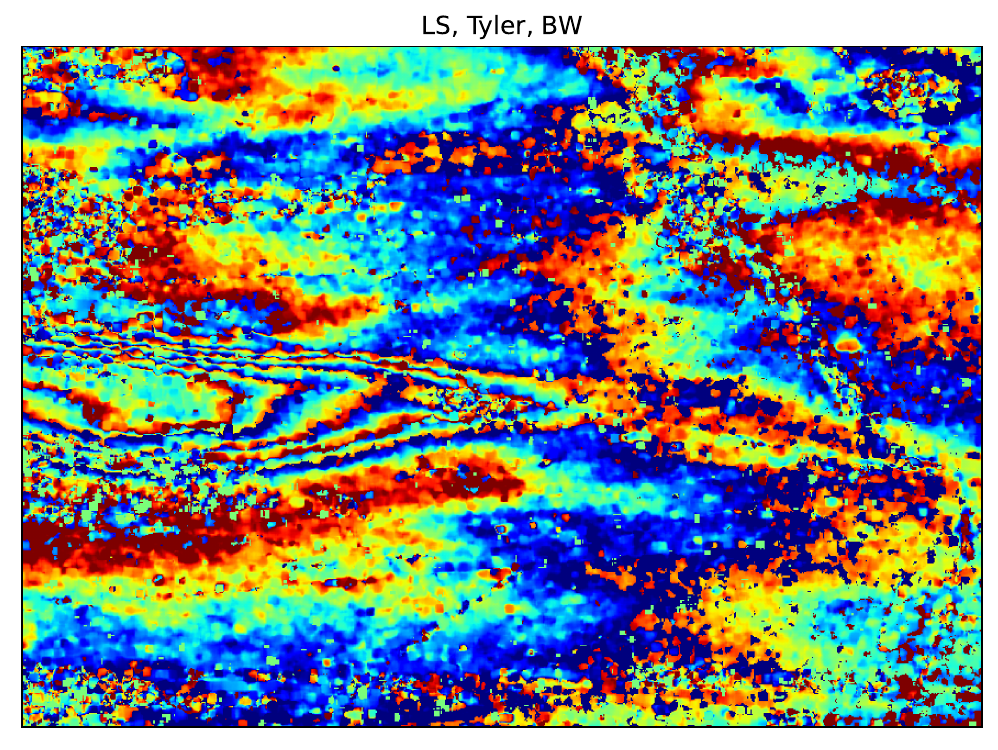}
    \end{minipage}
    \hfill{}
    \begin{minipage}{0.32\textwidth}
        \centering
        \includegraphics[width=1.0\columnwidth]{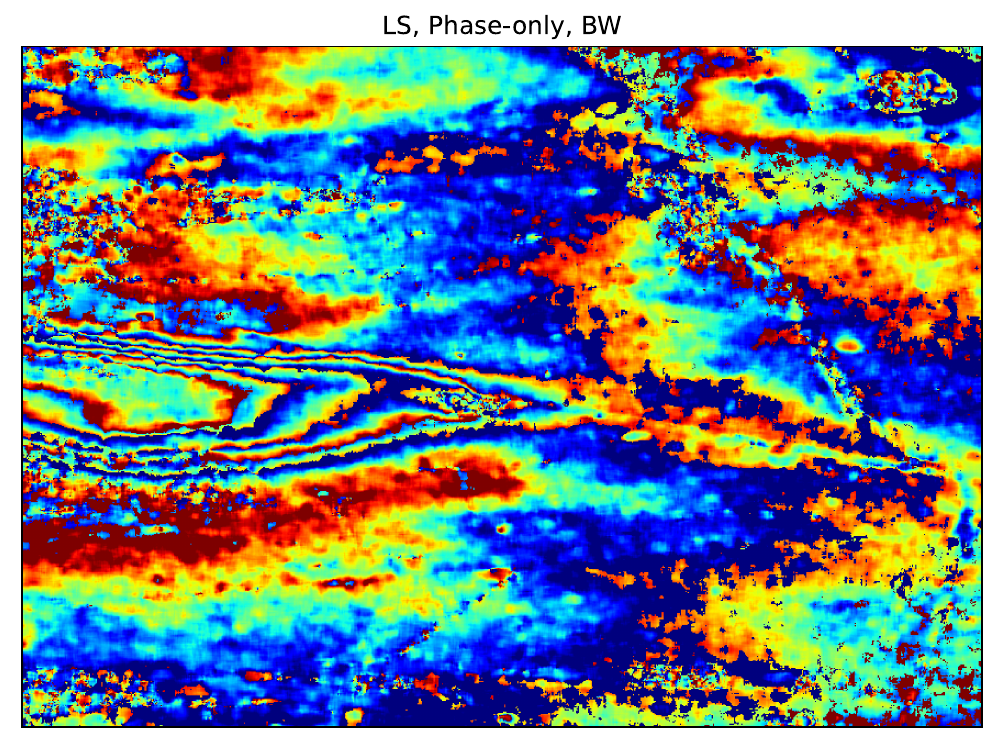}
    \end{minipage}
    \begin{minipage}{0.32\textwidth}
        \centering
        \includegraphics[width=1.0\columnwidth]{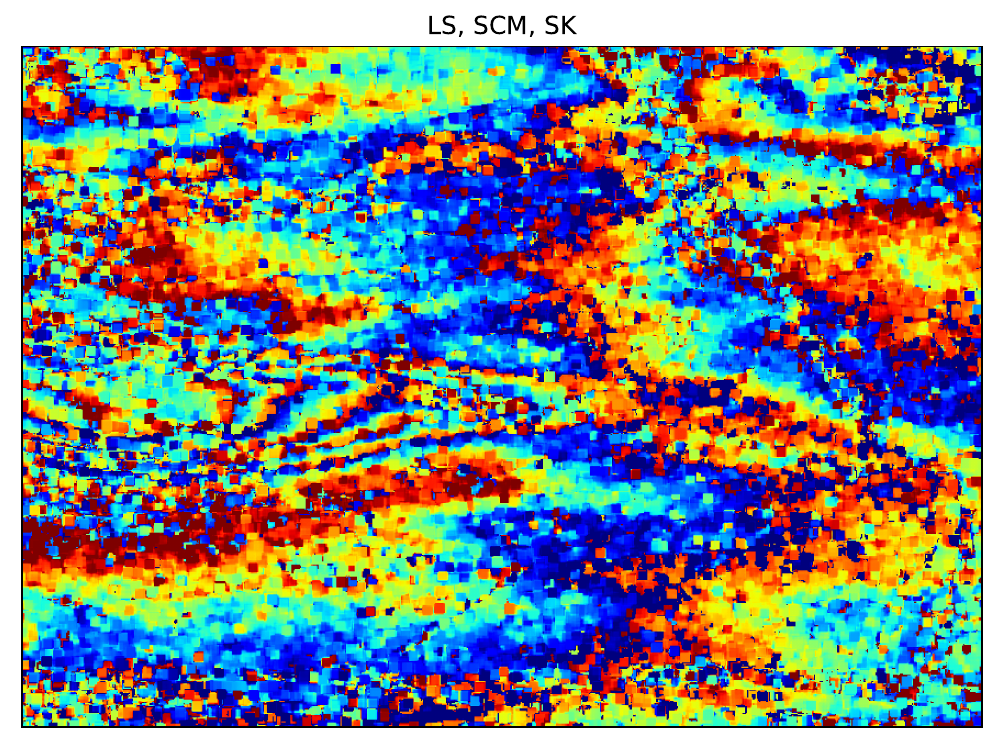}
    \end{minipage}
    \hfill{}
    \begin{minipage}{0.32\textwidth}
        \centering
        \includegraphics[width=1.0\columnwidth]{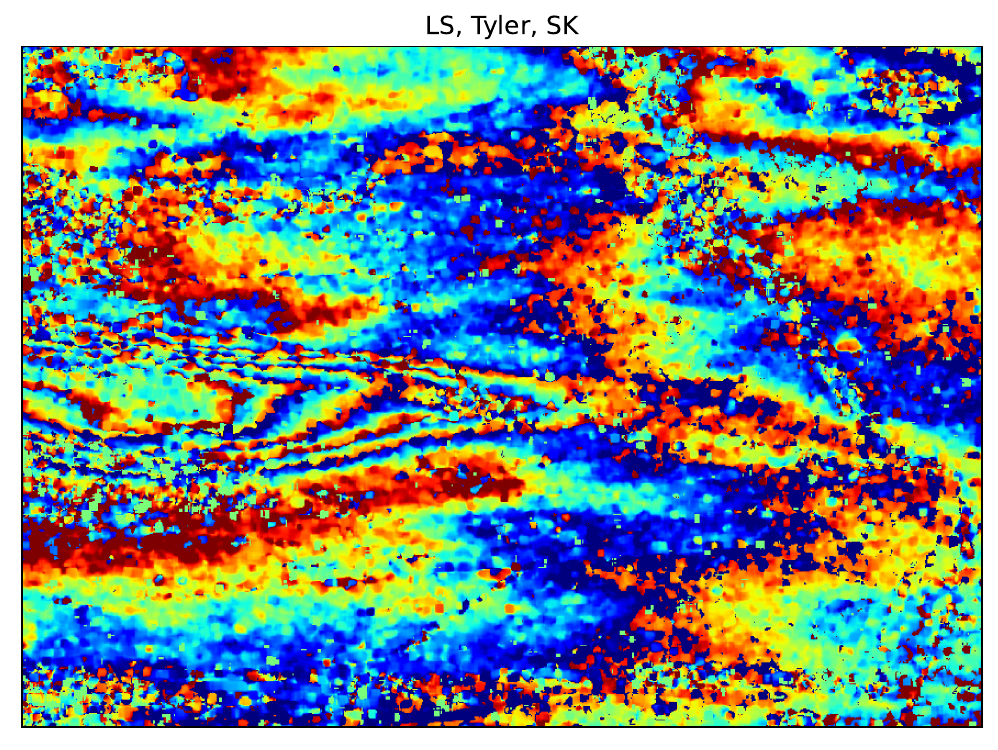}
    \end{minipage}
    \hfill{}
    \begin{minipage}{0.32\textwidth}
        \centering
        \includegraphics[width=1.0\columnwidth]{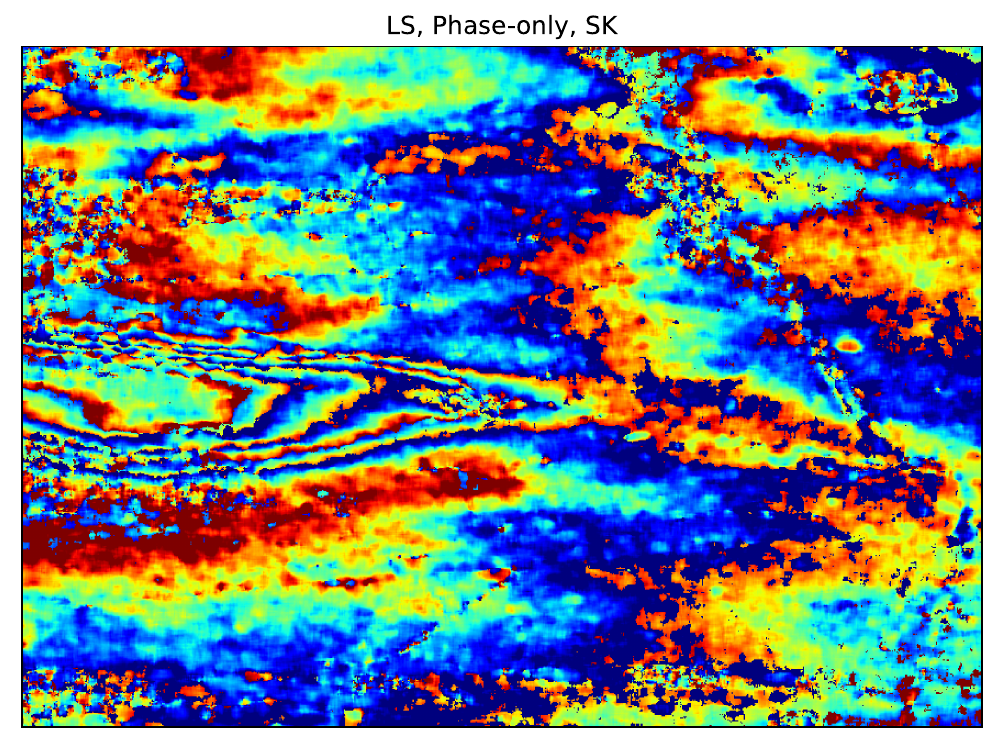}
    \end{minipage}   

    \begin{minipage}{1.0\linewidth}
    \centering
        \includegraphics[width=0.6\columnwidth]{./colorbar.pdf}
    \end{minipage}
\caption{
COFI-PL with LS fitting applied to Mexico InSAR dataset with various plug-in estimators (from left to right: sample covariance matrix \eqref{eq:SCM}, Tyler's estimator \eqref{eq:mest} with $u_{T}(t) = p/t$, phase-only sample covariance matrix \eqref{eq:phase_only_SCM}) and various regularization processes (from top to bottom: no regularization, low-rank approximation \eqref{eq:low_rank_approx2} with $k=1$, tapering \eqref{eq:tapering} with bandwidth $b=9$, shrinkage to identity \eqref{eq:LWshrink} with $\beta=0.1$).
} \label{fig:LS}
\end{figure*}

\begin{figure*}[htb]
    \begin{minipage}{0.32\textwidth}
        \centering
        \includegraphics[width=1.0\columnwidth]{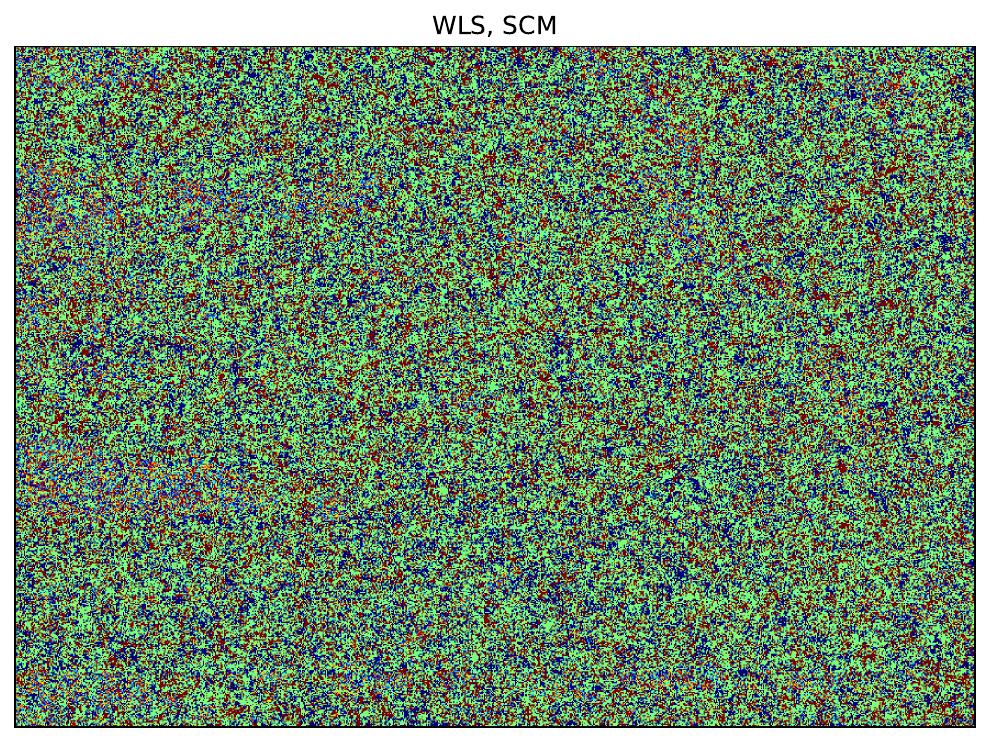}
    \end{minipage}
    \hfill{}
    \begin{minipage}{0.32\textwidth}
        \centering
        \includegraphics[width=1.0\columnwidth]{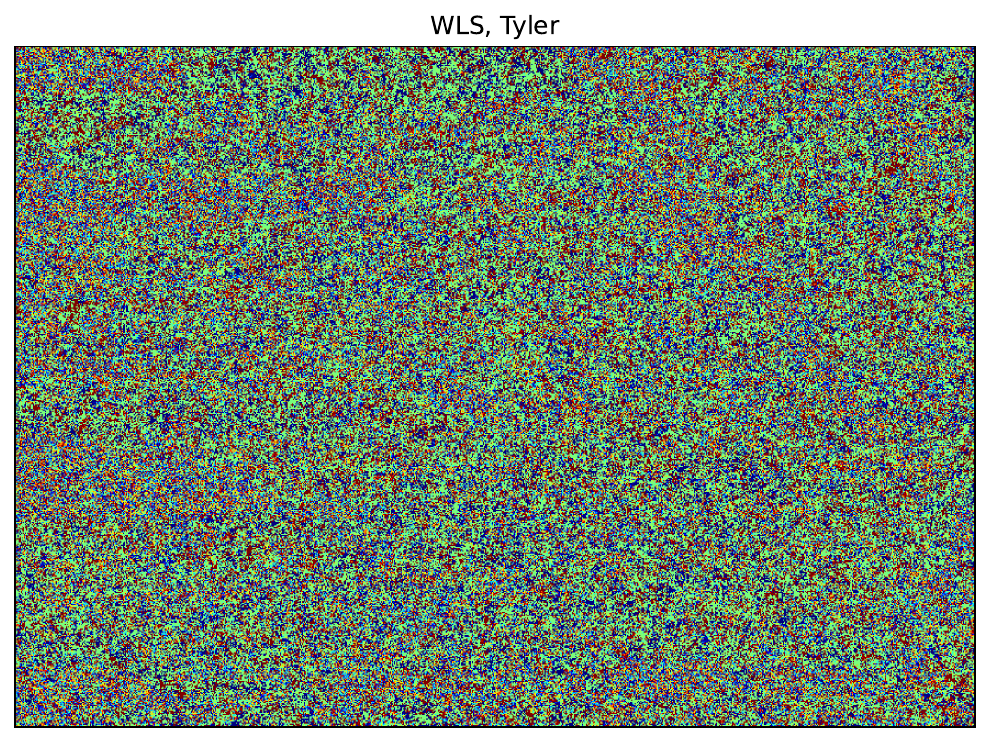}
    \end{minipage}
    \hfill{}
    \begin{minipage}{0.32\textwidth}
        \centering
        \includegraphics[width=1.0\columnwidth]{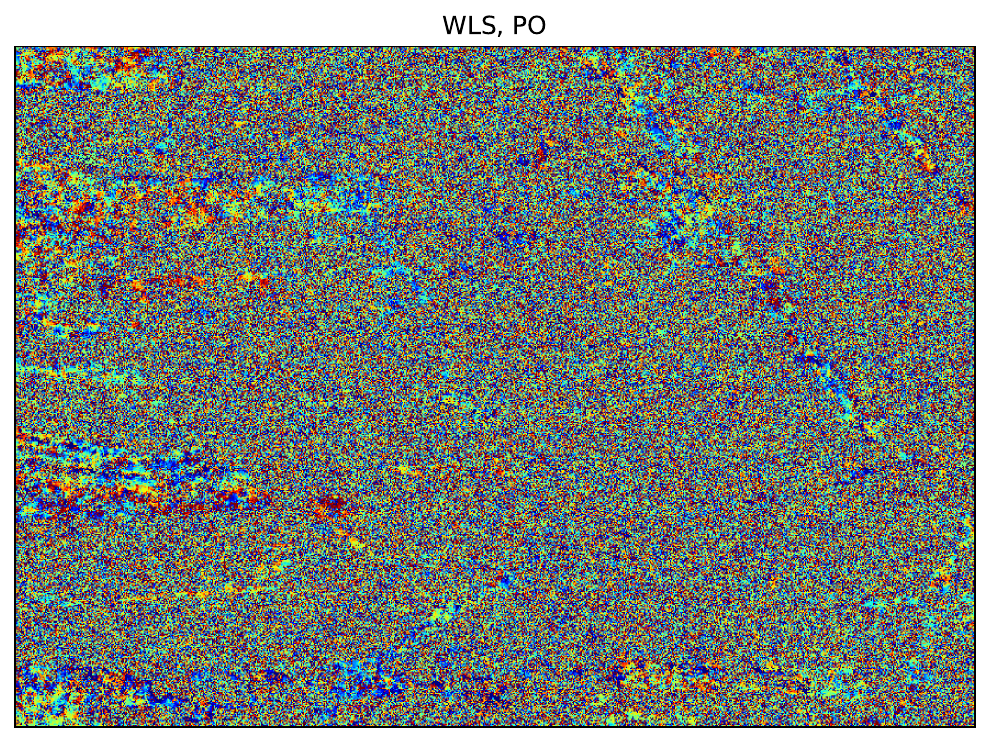}
    \end{minipage}
    \begin{minipage}{0.32\textwidth}
        \centering
        \includegraphics[width=1.0\columnwidth]{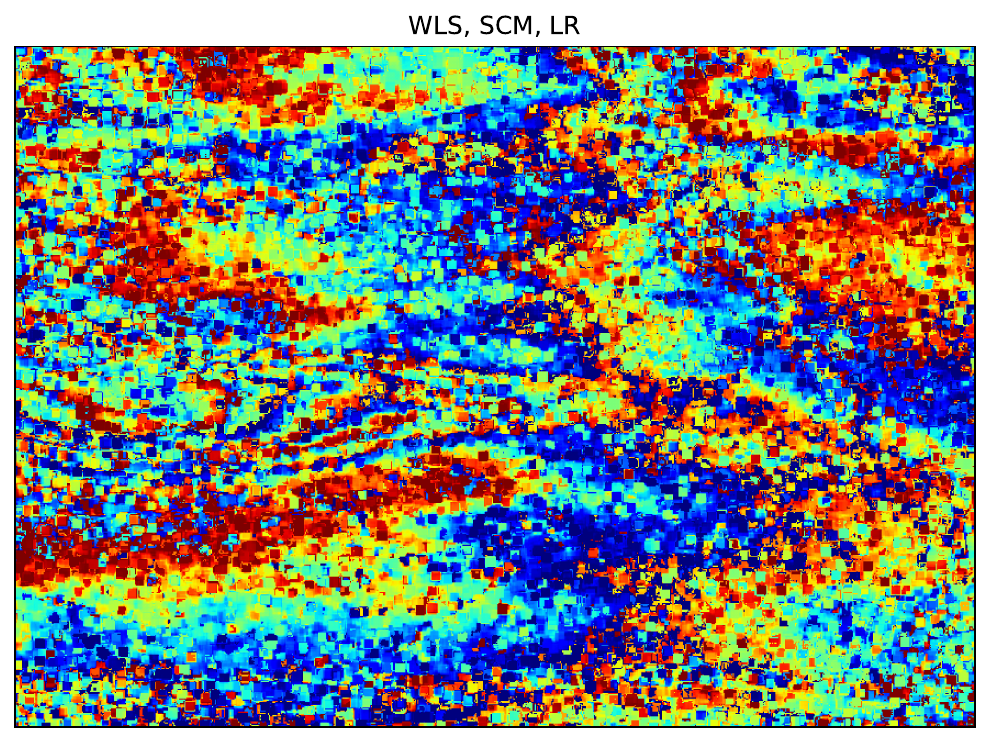}
    \end{minipage}
    \hfill{}
    \begin{minipage}{0.32\textwidth}
        \centering
        \includegraphics[width=1.0\columnwidth]{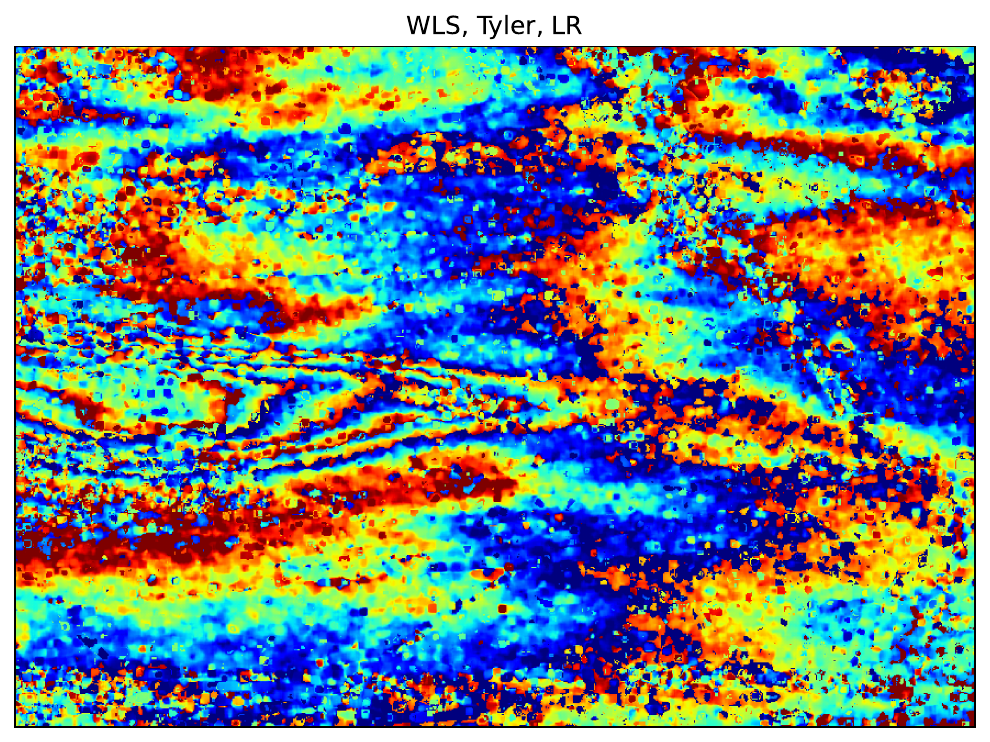}
    \end{minipage}
    \hfill{}
    \begin{minipage}{0.32\textwidth}
        \centering
        \includegraphics[width=1.0\columnwidth]{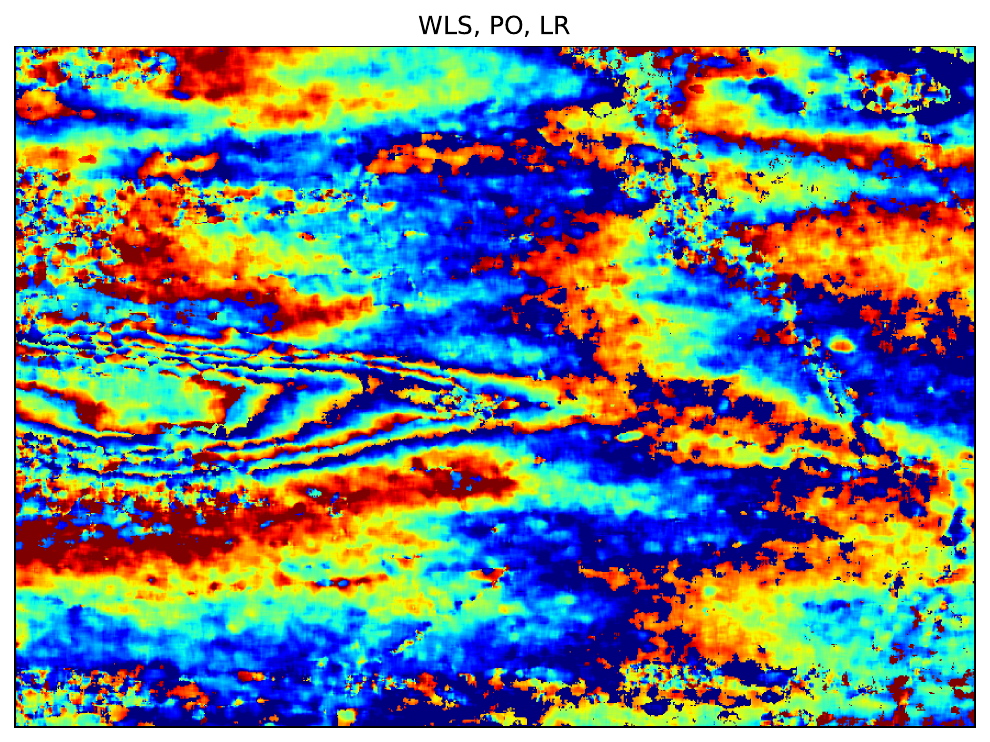}
    \end{minipage}
    \begin{minipage}{0.32\textwidth}
        \centering
        \includegraphics[width=1.0\columnwidth]{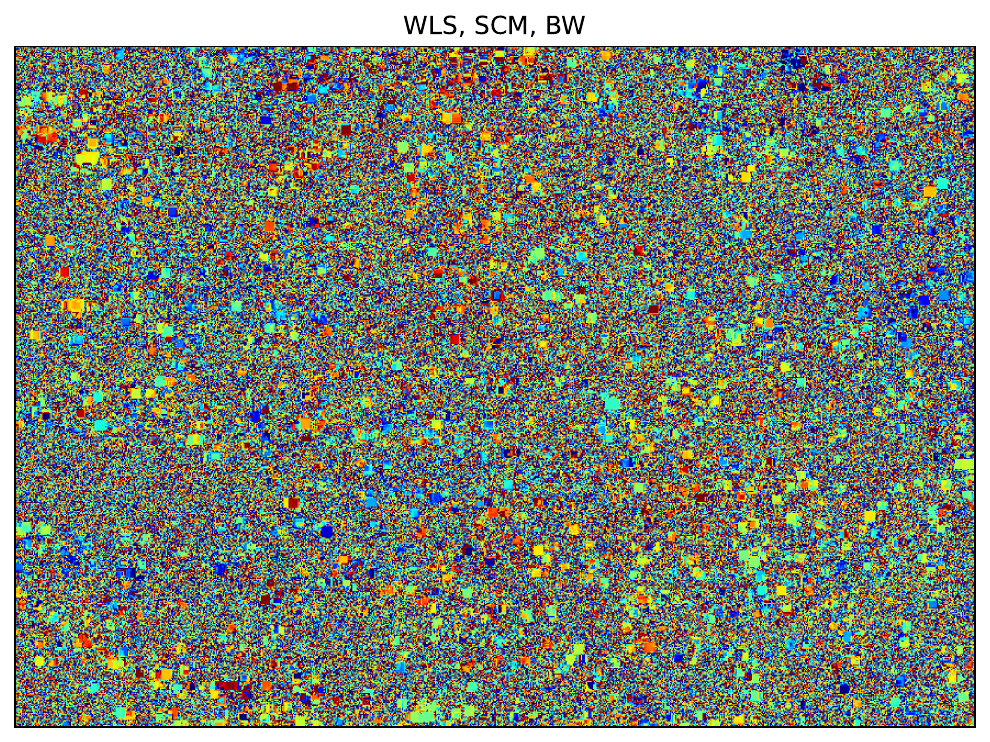}
    \end{minipage}
    \hfill{}
    \begin{minipage}{0.32\textwidth}
        \centering
        \includegraphics[width=1.0\columnwidth]{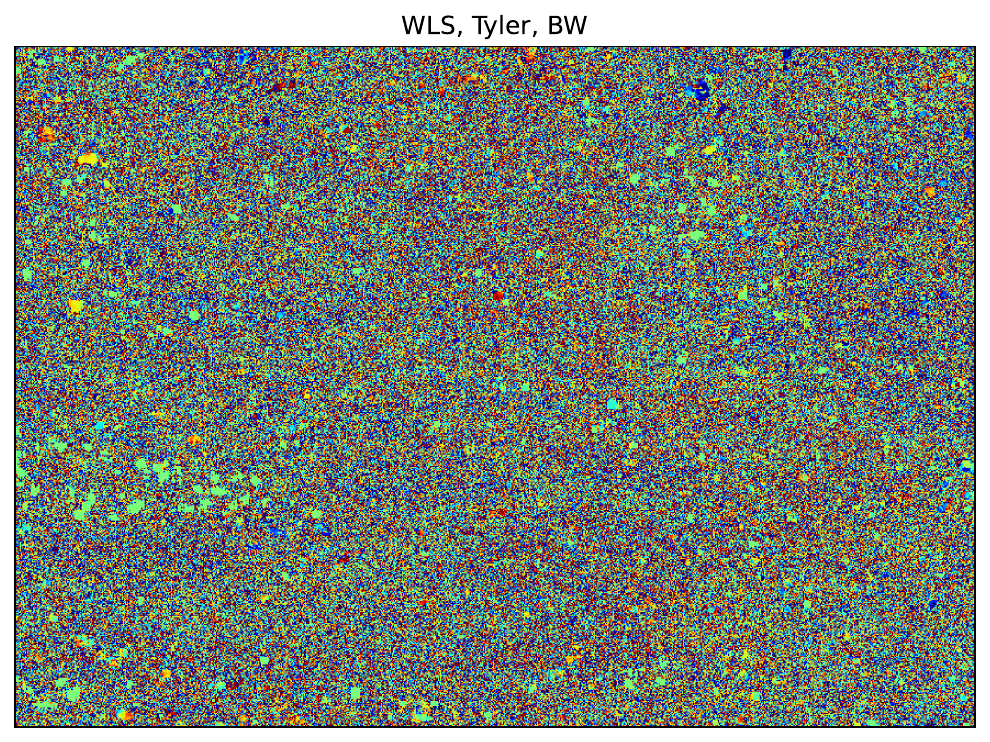}
    \end{minipage}
    \hfill{}
    \begin{minipage}{0.32\textwidth}
        \centering
        \includegraphics[width=1.0\columnwidth]{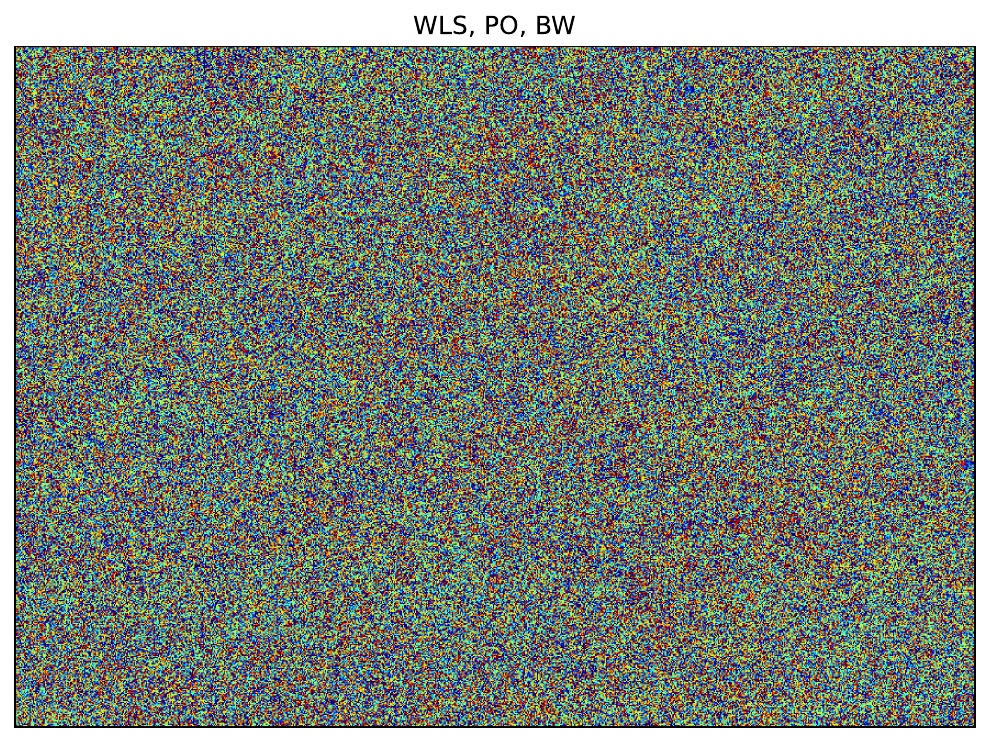}
    \end{minipage}
    \begin{minipage}{0.32\textwidth}
        \centering
        \includegraphics[width=1.0\columnwidth]{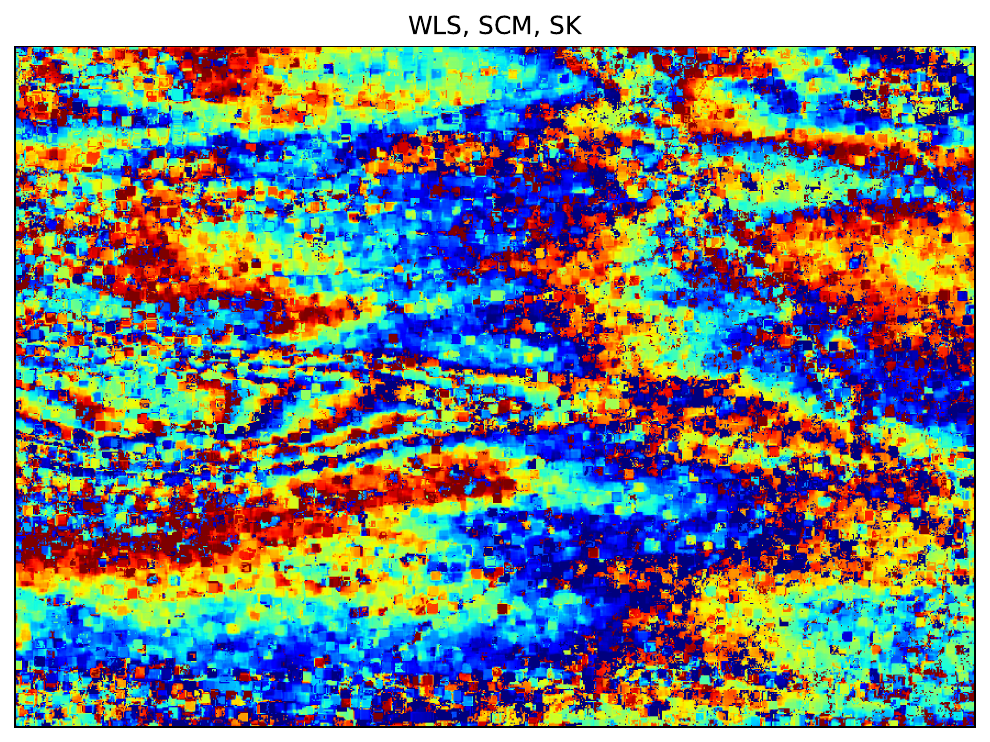}
    \end{minipage}
    \hfill{}
    \begin{minipage}{0.32\textwidth}
        \centering
        \includegraphics[width=1.0\columnwidth]{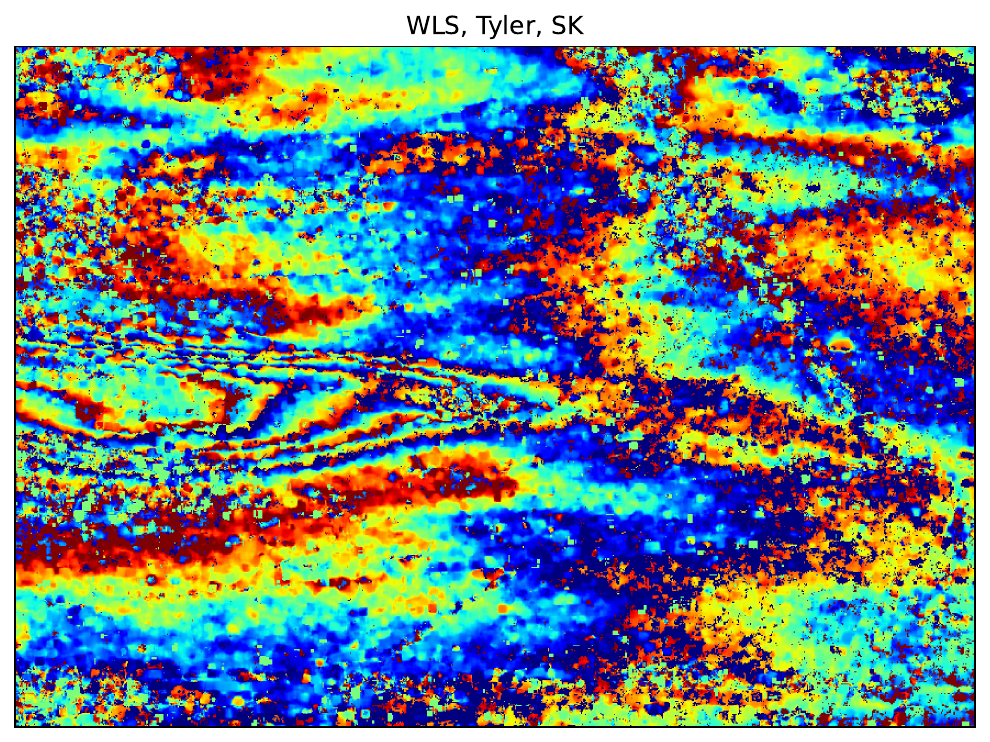}
    \end{minipage}
    \hfill{}
    \begin{minipage}{0.32\textwidth}
        \centering
        \includegraphics[width=1.0\columnwidth]{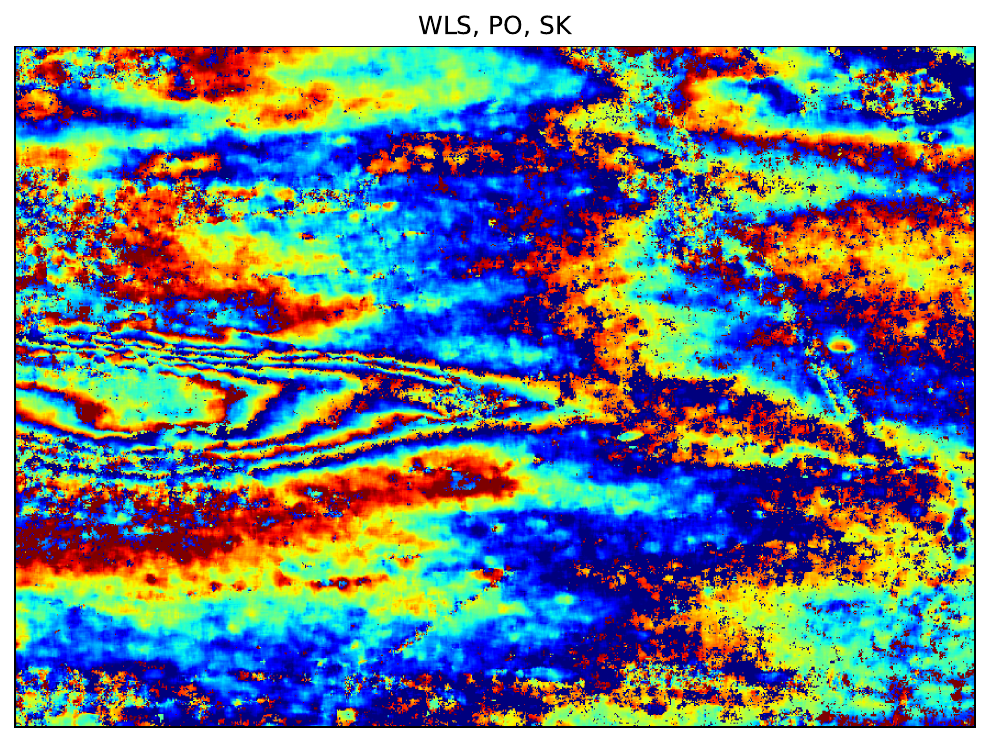}
    \end{minipage}  

    \begin{minipage}{1.0\linewidth}
    \centering
        \includegraphics[width=0.6\columnwidth]{./colorbar.pdf}
    \end{minipage}
\caption{
COFI-PL with WLS fitting applied to Mexico InSAR dataset with various plug-in estimators (from left to right: sample covariance matrix \eqref{eq:SCM}, Tyler's estimator \eqref{eq:mest} with $u_{T}(t) = p/t$, phase-only sample covariance matrix \eqref{eq:phase_only_SCM}) and various regularization processes (from top to bottom: no regularization, low-rank approximation \eqref{eq:low_rank_approx2} with $k=1$, tapering \eqref{eq:tapering} with bandwidth $b=9$, shrinkage to identity \eqref{eq:LWshrink} with $\beta=0.1$).
} \label{fig:WLS}
\end{figure*}

\subsection{Dataset}
To assess the performances of all proposed algorithms, we investigate one dataset with a temporal coverage of a year. The dataset consists of 31 Sentinel-1 SLC SAR images acquired between July 2019 and June 2020 over the southern region of Mexico Valley, within an endoreic basin surrounded by volcanic mountains. Pre-processing steps such as Sentinel-1 dedicated processing, coregistration, initial interferogram computation and filtering are done using the free software, SNAP, developed by ESA \cite{snap}. 
Results presented in the next paragraph are obtained using $p=31$ temporal samples and $n=64$ spatial samples (2D sliding window with size $8 \times 8$ pixels).

\subsection{Results}

For a given objective function set by the chosen matrix distance $d$, COFI-PL is applied to the data with various options of plug-in estimates and regularizations.
Fig. \ref{fig:KL} (resp. Fig. \ref{fig:LS}, and Fig. \ref{fig:WLS}) presents the output of the algorithm when KL (resp. LS, and WLS) is used as fitting objective.
COFI-PL with KL corresponds to most of standard IPL formulations (albeit the slight difference regarding regularization discussed in Section \ref{sec:sota}).
In this configuration, we can clearly see in Fig. \ref{fig:KL} that the quality of the plug-in estimate plays an important role (cf. first line), and that the phase-only sample covariance matrix \eqref{eq:phase_only_SCM} provides a great improvement compared to the standard sample covariance matrix.
The regularization also plays a crucial part in stabilizing the inversion of the modulus of the plug-in estimate (i.e., $\tilde{\boldsymbol{\Psi}}^{-1}$) that is required by construction with the KL fitting cost.
In this setup, the low-rank approximation \eqref{eq:low_rank_approx2} and shrinkage to identity \eqref{eq:LWshrink} appear to be the best regularization options in order to compensate a poor estimation or a poor conditioning of this quantity.
The tapering regularization using a hard threshold appears less stable, probably because it affects the aforementioned inversion step.
Hence, other (smoother) forms of tapering could be envisioned \cite{ollila2022regularized}.
The same conclusions can be drawn for COFI-PL with WLS fitting objective presented in Fig. \ref{fig:WLS}.
This was to be expected because the KL and WLS distances have a similar underlying construction that involves the inverse of the modulus of the plug-in estimate. 
We can still notice some differences: WLS is even more sensitive to the inversion instabilities induced by the tapering (since it depends quadraticly on $\tilde{\boldsymbol{\Psi}}^{-1}$), however it provides an interesting alternative to KL is some setups (notably when using Tyler's estimator as plug-in).
The quality and robustness of the plug-in estimate also plays a role for COFI-PL with LS fitting objective, as we can see an improvement brought by the phase-only sample covariance matrix in the first line of Fig. \ref{fig:LS}.
Interestingly, the LS setup appears much less sensitive to the conditioning of the plug-in estimate because no matrix inversion is required.
Hence results appear more stable with respect to the various regularization options.
A closer inspection of the outputs further reveals that regularization is still beneficial, notably when tapering is used.
As a matter of fact, LS fitting on a banded phase-only sample covariance matrix allowed to obtain the best visual results on this dataset.

In terms of optimization, majorization-minimization and Riemannian gradient descent yielded the same outputs for KL and LS fitting objectives, and WLS could only be evaluated with the Riemannian optimization framework.
In practice, majorization-minimization for KL and LS was experienced to be easier to tune (as no side parameters such as step size are involved) and provided solutions more quickly than a plain Riemannian gradient descent with constant step size.
Note that faster Riemannian optimization schemes could probably be obtained with more involved adaptive step size selection rules or alternate Riemannian optimization methods (e.g., Riemannian conjugate gradient or BFGS).
However, this last point goes beyond the scope of this paper.

\section{Conclusion}

This paper presented COFI-PL: a compact framework to design IPL algorithms based on the covariance fitting approach.
The formulation encompasses existing methods initially derived as approximate maximum likelihood, and allowed us to present some promising alternatives using various robust plug-in estimators, regularization methods, and fitting objective functions.
It also provides a clear and systematic framework to perform ablation studies, investigate the interest of alternates options in a single module, or assess the coupling effects when considering a joint estimation-regularization method of the covariance matrix.
Two methods to deal with the resulting optimization problems were then introduced: majorization-minimization and Riemannian optimization on the torus of phase-only complex vector.
Generic principles and tools were presented so that any cost optimization problem constructed within the COFI-PL framework can be addressed efficiently.
A real-world case study regarding the subsidience of Mexico city finally illustrated the interest of the proposed approach.
Experiments on this dataset notably evidenced the practical interest of using $i$) a LS fitting cost rather than the traditional KL divergence emanating from the maximum likelihood perspective; $ii$) the phase-only sample covariance matrix (a newly introduced covariance matrix estimate in the context of IPL) as plug-in estimate.

\footnotesize 
\bibliographystyle{IEEEbib}
\bibliography{references}

\end{document}